\documentclass[usletter,12pt]{article}

\usepackage{graphicx}
\usepackage{amsmath,amsthm,amssymb,mathrsfs,mathtools}
\usepackage[OT1]{fontenc} 
\usepackage{etoc}
\usepackage{float}
\usepackage{fancyhdr}
\def\subsectiontitle{}
\def\subsubsectiontitle{}
\usepackage{tikz}
\usetikzlibrary{automata, positioning}

\usepackage[noamsmath]{kpfonts}
\usepackage{zi4}

\usepackage{enumitem}
\usepackage[margin=1in]{geometry}

\usepackage[colorlinks,breaklinks=true,pagebackref]{hyperref}
\usepackage{breakcites}
\usepackage{xcolor}
\AtBeginDocument{%
	\hypersetup{%
		linkcolor={red!50!black},%
		citecolor={blue!50!black},%
		urlcolor={blue!80!black}%
	}%
}%

\usepackage[nameinlink,noabbrev,sort,capitalise]{cleveref}
\makeatletter
\def\ps@pprintTitle{%
	\let\@oddhead\@empty
	\let\@evenhead\@empty
	\def\@oddfoot{\emph{Very preliminary version}\hfill\emph{This draft: \today}}%
	\let\@evenfoot\@oddfoot}
\makeatother
\usepackage{etoolbox}
\patchcmd{\pprintMaketitle}
{\fi\hrule}% the second rule
{\fi\ifvoid\extrainfobox\else\unvbox\extrainfobox\par\vskip10pt\fi\hrule}
{}{}

\newsavebox\extrainfobox

\usepackage{comment}
\usepackage{multirow}
\usepackage{multicol}
\usepackage{subcaption}

\usepackage{indentfirst}

\allowdisplaybreaks

\let\oldfootnote\footnote
\renewcommand\footnote[1]{\oldfootnote{\hspace{.4mm}#1}}
\makeatletter
\renewenvironment{proof}[1][\proofname] {\par\pushQED{\qed}\normalfont\topsep6\p@\@plus6\p@\relax\trivlist\item[\hskip\labelsep\bfseries#1\@addpunct{.}]\ignorespaces}{\popQED\endtrivlist\@endpefalse}
\makeatother
\let\oldFootnote\footnote
\newcommand\nextToken\relax

\renewcommand\footnote[1]{%
	\oldFootnote{#1}\futurelet\nextToken\isFootnote}

\newcommand\isFootnote{%
	\ifx\footnote\nextToken\textsuperscript{,}\fi}
\usepackage{blkarray}
\usepackage{graphicx,pgfpages,epsfig,multirow,lscape}
\usepackage{hyperref}
\usepackage{epic}
\usepackage{xcolor}
\usepackage{setspace}
\usepackage{tikz}
\usetikzlibrary{arrows,decorations,decorations.pathreplacing,calc,matrix}
\usepackage{natbib}
\DeclareFontFamily{U}{mathb}{\hyphenchar\font45}
\DeclareFontShape{U}{mathb}{m}{n}{
	<-6> mathb5 <6-7> mathb6 <7-8> mathb7
	<8-9> mathb8 <9-10> mathb9
	<10-12> mathb10 <12-> mathb12
}{}
\DeclareSymbolFont{mathb}{U}{mathb}{m}{n}
\DeclareMathSymbol{\llcurly}{\mathrel}{mathb}{"CE}
\DeclareMathSymbol{\ggcurly}{\mathrel}{mathb}{"CF}

\hypersetup{colorlinks = true}
\hypersetup{allcolors=blue}

\newtheorem{definition}{Definition}

\newtheorem{theorem}{Theorem}
\newtheorem*{theorem*}{Theorem}
\newtheorem{proposition}{Proposition}
\newtheorem{lemma}{Lemma}
\newtheorem{example}{Example}
\newtheorem{corollary}{Corollary}

\newtheorem{innercustomex}{}
\newenvironment{customex}[1]
{\renewcommand\theinnercustomex{#1}\innercustomex}
{\endinnercustomex}

\def\citeapos#1{\citeauthor{#1}'s (\citeyear{#1})}

\def\w{\omega}
\def\bw{\bar{\omega}}

\def\C{\mathcal{C}}
\def\O{\mathcal{O}}

\begin{document}

\title{Cores in discrete exchange economies with complex endowments\thanks{First draft: March 2020. I am grateful to the editor and two anonymous referees for comments. I acknowledge the financial support by National Natural Science Foundation of China (Grant 71903093 and Grant 72033004). Other acknowledgments are to be added.
}}

\author{Jun Zhang\\
	\small 
	Institute for Social and Economic Research, Nanjing Audit University, China\\ \small zhangjun404@gmail.com}

\date{April 28, 2021}

\maketitle

\begin{abstract}
    The core is a traditional and useful solution concept in economic theory. But in discrete exchange economies without transfers, when endowments are complex, the core may be empty. This motivates \cite{balbuzanov2019endowments} to interpret endowments as exclusion rights and propose a new concept called exclusion core. Our contribution is twofold. First, we propose a rectification of the core to solve its problem under complex endowments. Second, we propose a refinement of Balbuzanov and Kotowski's exclusion core to improve its performance. Our two core concepts share a common idea of correcting the misused altruism of unaffected agents in blocking coalitions. We propose a mechanism to find allocations in the two cores.
\end{abstract}



\thispagestyle{empty}
\setcounter{page}{0}
\newpage

\section{Introduction}

We study a model of discrete exchange economies where side payments are forbidden and endowments are complex. It is a generalization of \citeapos{shapley1974cores} housing market model and \citeapos{HZ1979} house allocation model. The two models are well studied in the market design literature. 
But, as \cite{balbuzanov2019endowments} (BK hereafter) recently emphasize, the endowment structures in the two models are too simple to describe complex property rights in practice.\footnote{In the housing market model, every agent owns a distinct object, while in the house allocation model, agents collectively own all objects. See \cite{balbuzanov2019endowments} for examples of complex property rights in practice that cannot be described by the two simple models.} A model that places complex endowments at the forefront is of theoretical interest and of practical relevance. In the model we study, which is attributed to BK,  an agent may own multiple objects, none at all, or co-own objects with the others in a complex manner. 

In models with simple endowments, desirable allocations are well defined and found by mechanisms such as top trading cycle and serial dictatorship. But when endowments are complex, it is not immediately clear what kind of allocations are desirable and reflect agents' complex property rights. Given this, as a traditional and useful concept in economic theory, the core becomes a natural solution. The core predicts the set of allocations that will appear when agents can freely form coalitions to improve upon their assignments by reallocating their endowments among themselves. Also, if we recall its success and relation with desirable mechanisms in models with simple endowments, we may expect that the core can hint at desirable mechanisms to solve our model.\footnote{In the housing market model, the core is a singleton and coincides with the outcome of the top trading cycle mechanism \citep{roth1977weak}. In the house allocation model, the core equals the set of Pareto efficient allocations, which are the outcomes of the serial dictatorship mechanism. In a generalized indivisible goods allocation model, \cite{sonmez1999strategy} further relates the existence of individually rational, Pareto efficient and strategy-proof mechanisms to the single-valuedness of the core.} However, an unpleasant fact is that, the core may be empty under complex endowments.\footnote{To be precise, there are two definitions of the core in discrete exchange economies without transfers. One is based on weak domination and called the \textit{strong core}, and the other is based on strong domination and called the \textit{weak core}. We will define the two notions in the paper, but in Introduction we use ``the core'' to mean the strong core.} 

Our \textbf{first} contribution is to find the root of the problem with the core under complex endowments and to rectify it. In the definition of the core, a group of agents can form a coalition to block an allocation if by reallocating their endowments among themselves, they are no worse off and some of them are strictly better off. Assuming that unaffected agents are willing to join a blocking coalition to help the others in the coalition without harming themselves is crucial for the core to be efficient, and this is often called altruism. But we find that this altruism argument is problematic or misused under complex endowments. Our finding is related to \citeauthor{shapley1974cores}'s observation that, in the housing market model if agents' preferences are not strict, the core may be empty.  In our model although agents' preferences are strict, when a group of agents owns more endowments than their demands and they are satisfied by some of their endowments, they will be indifferent about the allocations of their remaining endowments, but they have the rights to determine the allocations of their remaining endowments. With this resulting indifferent preferences, they may join too many blocking coalitions to make the core empty. This is much like \citeauthor{shapley1974cores}'s observation in the housing market model.

We ascribe the above problem to the altruism argument, because we find that there is a fundamental difference between complex endowments and simple endowments regarding the incentive of unaffected agents to join a blocking coalition. We call a group of unaffected agents \textbf{self-enforcing} if their assignments come from their endowments. Such a group can ensure their assignments without relying on any others. In our model, when a self-enforcing group of unaffected agents has more endowments than their assignments, if they are willing to help the others obtain their remaining endowments, this is regarded as pure altruism. Intuitively, pure altruism should be neutral to the others. That is, there is no reason to expect that self-enforcing unaffected agents want to be biased towards some others. But the definition of the core admits biased altruism: self-enforcing unaffected agents are willing to join a blocking coalition to help those in the coalition and harm those outside the coalition. By contrast, the traditional altruism argument does not result in a problem in the housing market model, because self-enforcing unaffected agents in the model are never purely altruistic. For any group of unaffected agents in a blocking coalition, if they are self-enforcing, they must contribute nothing to the others in the coalition. If they are not self-enforcing, it means that they rely on the others in the coalition to ensure their assignments. So their participation in the coalition cannot be regarded by pure altruism. This explains why their so-called altruism may be biased and favor those in the coalition.

After identifying the root of the problem with the core, we rectify its definition by adding a neutrality assumption on self-enforcing unaffected agents. When a blocking coalition includes self-enforcing unaffected agents, we require that they cannot reallocate any of their endowments that have been allocated to the agents outside the coalition to help those in the coalition. We call the concept in the new definition \textbf{rectified core}. The rectified core supersedes the core, and equals the core in many special cases (Section \ref{Section:special:economies}). Most importantly, it is efficient and nonempty.

Our \textbf{second} contribution is to propose a refinement of BK's solution. After observing the problem with the core, BK choose a non-standard interpretation of endowments: agents have the right to evict any others who occupy their endowments, and this exclusion right can be extended to the endowments of those who occupy their endowments, and so on. To see its difference with the standard interpretation that agents have the right to consume or exchange endowments, let us consider a simple example in which two agents $ \{1,2\} $ co-own an object $ a $ and another agent $ 3 $ owns an object $ b $. Suppose agents $ 1 $ and $ 3 $ prefer $ a $ to $ b $, and agent $ 2 $ prefers $ b $ to $ a $. In an allocation where $ a $ is assigned to $ 3 $ and $ b $ is assigned to $ 1 $, by BK's definition, $ \{1,2\} $ get direct control of their endowment $ a $ and indirect control of $ b $, because $ b $ is the endowment of the agent who occupies $ a $. So $ \{1,2\} $ are able to evict $ 3 $ from $ a $ to reach the allocation where $ a $ is assigned to $ 1 $ and $ b $ is assigned to $ 2 $. Both $ 1 $ and $ 2 $ become better off. This blocking cannot happen in the standard interpretation of endowments because $ \{1,2\} $ cannot make $ 3 $ worse off but occupy $ 3 $'s endowment. BK define the \textbf{exclusion core} as the set of allocations where no coalition can make themselves better off by evicting the others from objects they directly or indirectly control. The exclusion core is nonempty and rules out many unintuitive allocations.

BK emphasize that excluding unaffected agents from blocking coalitions they define is crucial for the exclusion core to be nonempty. But they also notice that in some economies the exclusion core fails to rule out unintuitive allocations that the standard interpretation of endowments can easily rule out. For example, suppose two agents $ \{1,2\} $ co-own an object $ a $, but $ a $ is assigned to another agent $ 3 $ who owns nothing, while nothing is assigned to $ 1 $ and $ 2 $. Because $ \{1,2\} $ cannot be better off simultaneously, they cannot evict $ 3 $ from $ a $, and this unintuitive allocation belongs to the exclusion core. BK turn to a model with more intricate exclusion rights to rule out this unintuitive allocation. Differently, we solve this inadequacy of the exclusion core by allowing unaffected agents to join blocking coalitions. Like our idea for the rectified core, we modify the altruism of self-enforcing unaffected agents to require that if they participate in a blocking coalition, their purpose is only to help the others in the coalition use their joint exclusion rights; they never harm the others outside the coalition by using their own exclusion rights.\footnote{When unaffected agents are self-enforcing, they cannot be evicted by any others, and so their participation in an exclusion blocking coalition can be regarded as altruism.} In the above example, $ \{1,2\} $ will be able to evict $ 3 $ and assign $ a $ of one of them. The other agent in the coalition is unaffected and self-enforcing (by obtaining nothing), but the eviction reflects their joint exclusion right. The \textbf{refined exclusion core} we propose remains to be nonempty.

The rectified core and the refined exclusion core are independent concepts and based on different interpretations of endowments. But they share a common idea of correcting the misused altruism of unaffected agents in blocking coalitions. We treat our proposal of the rectified core as a response to BK by showing that dropping the standard interpretation of endowments is unnecessary if the purpose is only to solve the problem with the core. On the other hand, we acknowledge that BK's new interpretation of endowments is of independent merit and has different power than the standard interpretation. Interestingly, the two interpretations are compatible. We prove that the rectified core and the refined exclusion core do not include each other, but they have a nonempty intersection. We propose a mechanism whose outcomes belong to their intersection. So we may regard it as a desirable mechanism to solve our model. It is a generalization of the so-called ``you request my house - I get your turn'' mechanism, which is originally proposed by \cite{AbduSonmez1999} to solve a hybrid of the house allocation model and the housing market model. 

To further explore the  relations between our solutions and those they are to replace, we analyze several special cases of our model. These special cases are also of independent interest because they are natural generalizations of the models with simple endowments in the literature. Some of the solutions will become equal in these special cases. In particular, if every group of agents never has more endowments than their endowments, then the core will equal the rectified core, which means that the core is nonempty.  

We organize the paper as follows. After presenting the model in Section \ref{Section:model}, we analyze the problem with the core and present our rectification in Section \ref{Section:rectified:core}. We then propose our refinement of BK's exclusion core in Section \ref{Section:refine:exclusion core}. We present our mechanism in Section \ref{Section:YRMH}, and clarify the relations between the several solutions in Section \ref{Section:relations}. We further study several special types of economies in Section \ref{Section:special:economies}. We discuss two applications in Section \ref{Section:application}, and discuss related literature in Section \ref{Section:discussion}. Appendix includes all proofs.

\section{The complex endowments model}\label{Section:model}

In the complex endowments model of \cite{balbuzanov2019endowments},\footnote{BK call the model here \textit{simple economies}. They also consider an extension called \textit{relational economies}. See the related literature in Section \ref{Section:discussion}.} an economy is represented by a tuple $ \Gamma=(I,O,\succ_I,\{C_o\}_{o\in O}) $, where $ I $ is a finite set of agents, $ O $ is a finite set of indivisible heterogeneous objects, $ \succ_I $ is a preference profile of agents, and $ \{C_o\}_{o\in O} $ is the endowment system. We use $ i $ or $ j $ to denote an agent, and use $ a $, or $ b $, or $ o $ to denote an object. Every $ i $ demands at most one object, and has strict preferences represented by a linear order $ \succ_i $ on $ O\cup \{o^*\} $, where $ o^* $ is a null object. An object $ o $ is acceptable to $ i $ if $ o\succ_i o^* $. For any two objects $ a$ and $b $, we write $ a\succsim_i b $ if $ a=b $ or $ a\succ_i b $. Every nonempty subset of agents is called a \textbf{coalition}. A coalition $ C' $ is a \textbf{sub-coalition} of another coalition $ C $ if $ C'\subsetneq C $.  Every $ o $ is owned by a nonempty subset of agents $ C_o\subset I$. Every $ i\in C_o $ is called an owner of $ o $. If $ C_o $ is a singleton, $ o $ is called privately owned by the agent in $ C_o $, while if $ C_o=I $, $ o $ is called publicly owned by all agents. The set of \textbf{endowments} owned by each coalition $ C $, denoted by $ \w(C) $, consists of objects that are owned by $ C $ or by sub-coalitions of $ C $. So $ \w(C)=\{o\in O:C_o\subset C\} $.

An \textbf{allocation} is a mapping $ \mu:I\rightarrow O\cup \{o^*\} $ such that $ |\mu(i)|=1 $ for all $ i\in I $ and $ |\mu^{-1}(o)|= 1 $ for all $ o\in O $. For every $ i $, $ \mu(i) $ is the object assigned to $ i $, and if $ \mu(i)=o^* $, it means that $ i $ receives nothing. For every coalition $ C$, let $ \mu(C)=\cup_{i\in C}\mu(i) $ denote the set of objects assigned to $ C $.  An allocation $ \mu $ is Pareto dominated by another allocation $ \sigma $ if $ \sigma(i)\succsim_i \mu(i) $ for all $ i$, and $ \sigma(j)\succ_j \mu(j) $ for some $ j$. An allocation is \textbf{Pareto efficient} if it is not Pareto dominated by any other allocation. We use \textbf{PE} to denote the set of Pareto efficient allocations in any economy.

Let $ \mathcal{E} $ denote the set of economies and $ \mathcal{A} $ the set of allocations. A solution is a correspondence $ f:\mathcal{E}\rightarrow 2^\mathcal{A}$, which finds a set of allocations $ f(\Gamma) $ for every $ \Gamma\in \mathcal{E} $. Note that we allow $ f(\Gamma) $ to be empty for some $ \Gamma $.

The house allocation model and the housing market model are two special cases of the complex endowments model. In the house allocation model, all objects are publicly owned. In the housing market model, there are equal numbers of agents and objects, and every agent owns a distinct object. In both models, it is often assumed that agents regard all objects as acceptable.

\section{Rectification of the core}\label{Section:rectified:core}

\subsection{The problem with the core under complex endowments}

In discrete exchange economies without transfers, \cite{shapley1974cores} and \cite{roth1977weak} have observed that it is necessary to distinguish between two different notions of the core: the weak core and the strong core. An allocation belongs to the weak core if no coalition can make all members strictly better off by reallocating their endowments among themselves. An allocation belongs to the strong core if no coalition can make all members no worse off and at least one member strictly better off by reallocating their endowments among themselves.

\begin{definition}\label{Definition:weakblock}
	An allocation $ \mu $ is \textbf{weakly blocked} by a coalition $ C $ via another allocation $ \sigma $ if
	\begin{enumerate}
		\item $ \forall i\in C $, $ \sigma(i) \succsim_i \mu(i) $ and $ \exists j\in C $, $ \sigma(j)\succ_j\mu(j) $;
		
		\item $ \sigma(C)\subset \w(C)\cup \{o^*\}$.
	\end{enumerate}
	In the above definition, if $\forall i\in C $, $ \sigma(i) \succ_i \mu(i) $, then $ \mu $ is \textbf{strongly blocked} by $ C $ via $ \sigma $.

	The \textbf{strong core} consists of allocations that are not weakly blocked, and the \textbf{weak core} consists of allocations that are not strongly blocked.
\end{definition}

We call a coalition $ C $ \textbf{self-enforcing} in an allocation $ \sigma $ if $ \sigma(C)\subset \w(C)\cup \{o^*\} $. A self-enforcing coalition can ensure their assignments by making agreements among themselves, without relying on those outside the coalition. The two blocking notions in Definition \ref{Definition:weakblock} require a blocking coalition to be self-enforcing, but they make different assumptions on agents' desire to enforce blocking. In strong blocking, agents are selfish. They are willing to join a blocking coalition if and only if they benefit from the blocking. While in weak blocking, agents are willing to help the others if they are not harmed. The participation of unaffected agents in a blocking coalition is often called altruism. Since \cite{shapley1974cores}, it has been well understood that the weak core may be inefficient, and the participation of unaffected agents in weak blocking is crucial for the strong core to be efficient.\footnote{If an allocation $ \mu $ is Pareto dominated by another allocation $ \sigma $, then $ \mu $ is weakly blocked by the grand coalition $ I $ via $ \sigma $.}  So the strong core is more appealing than the weak core.

In the housing market model, the strong core is a singleton. In the house allocation model, the strong core equals the set of Pareto efficient allocations. However, when we allow for complex endowments defined in Section \ref{Section:model}, the strong core may be empty. See the following  example of BK.

\begin{example}\label{Example:Kingdom}
	Consider three agents $ \{1,2,3\} $ and two objects $ \{a,b\} $. Both objects are privately owned by agent $ 1 $. All agents regard both objects as acceptable and prefer $ a $ to $ b $. We use the following tables to represent endowments, preferences, and three allocations $ \{\mu,\sigma, \delta\}$ under our consideration.
	\begin{center}
		\begin{tabular}{ccc}
			& $ a $ & $ b $  \\ \hline
			$ C_o $: & $ 1 $ & $ 1 $ \\ \hline
			$ \mu $: & $ 1 $ &  \\
			$ \sigma $: & $ 1 $ & $ 2 $\\
			$ \delta $: & $ 1 $ & $ 3 $ 
		\end{tabular}
		\quad \quad
		\begin{tabular}{ccc}
			$ \succ_1 $	& $ \succ_2 $ & $ \succ_3 $  \\ \hline
			$ a $ & $ a $ & $ a $ \\
			$ b $ & $ b $ & $ b $ \\
			& \\
			&
		\end{tabular}
	\end{center}
	
	Because $ 1 $ owns all objects, it is intuitive that $ 1 $ should obtain the best object $ a $. Given that, if we do not want $ b $ to be wasted, one of the other two agents should obtain $ b $. So $ \sigma $ and $ \delta $ are intuitive allocations in this example. Both $ \sigma $ and $ \delta $ Pareto dominate $ \mu $, in which $ b $ is wasted.
	
    Now we examine the two core notions. The weak core in this example equals $ \{\mu,\sigma, \delta\}$. If $ 1 $ does not obtain $ a $, $ 1 $ alone will strongly block the allocation. Yet though $ \mu $ is inefficient, it cannot be strongly blocked because any strong blocking coalition has to include $ 1 $, but because $ 1 $ cannot be improved she is not willing to join any coalition. 
	
	However, the strong core is empty. Clearly, $ \mu $ is weakly blocked by the grand coalition via $ \sigma $ or $ \delta $. What strange is that, $ \sigma $ and $ \delta $ can be weakly blocked via each other: $ \{1,3\} $ can weakly block $ \sigma $ via $ \delta $ by reallocating $ b $ from $ 2 $ to $ 3 $, while $ \{1,2\} $ can weakly block $ \delta $ via $ \sigma $ by reallocationg $ b $ from $ 3 $ to $ 2 $. Because $ 1 $ owns all objects, the two coalitions are self-enforcing, and because $ 1 $ is unaffected, she is willing to join each of the two coalitions to  help the other agent obtain $ b $.
\end{example}

Example \ref{Example:Kingdom} reminds us of the early observation of \cite{shapley1974cores} that, in the housing market model, if agents' preferences are not strict, then the strong core may be empty. They present a simple example where three agents $ 1,2,3 $ respectively own three objects $ a,b,c $. $ 1 $ and $ 3 $ most prefer $ b $ and then prefer their own endowment to the other's endowment, while $ 2 $ is indifferent between all objects. There are two Pareto efficient allocations. In one allocation, $ 1 $ and $ 2 $ exchange endowments, and $ 3 $ remains with her endowment. In the other allocation, $ 2 $ and $ 3 $ exchange endowments, and $ 1 $ remains with her endowment. In each of the two allocations, $ 2 $ is willing to form a coalition with the agent who does not obtain $ b $ to reallocate $ b $ to the agent. So the strong core is empty.

What happens in Example \ref{Example:Kingdom} is very similar. After agent $ 1 $ obtains her favorite object $ a $, which is one of her endowments, she is indifferent about the allocation of her remaining endowment $ b $, but she has the right to determine it. Under the altruism assumption  in weak blocking, $ 1 $ is willing to form a coalition with the agent who does not obtain $ b $ to reallocate $ b $ to the agent. This makes the strong core empty. In general, although agents have strict preferences over objects in our model, when a group of agents owns more endowments than their demands, indifferent preferences over allocations can implicitly appear to make the strong core empty. Section \ref{Section:special:economies} proves that if no group of agents owns more endowments than their demands, the strong core must be nonempty.

The contrast between the failure of the strong core under complex endowments and its success under simple endowments motivates us to compare the altruism of unaffected agents in the two environments. In the house allocation model, because all objects are public endowments, every blocking coalition must include all agents, so that the blocking must be a Pareto improvement. While in the housing market model, although the participation of unaffected agents in a blocking coalition is often explained by altruism, we argue that this altruism is not pure. To see this, suppose that a blocking coalition $ C $ includes a group of unaffected agents $ C' $. If $ C' $ is self-enforcing, then their endowments must be fully allocated to themselves so that they contribute nothing to the others in $ C $. So it is without loss to remove $ C' $ from $ C $. Then we can assume that every group of unaffected agents $ C' $ in  $ C $ is not self-enforcing. This means that $ C' $ relies on the others in $ C $ to ensure their assignments. So their participation in $ C $ can be due to the concern that they may  lose their assignments if they do not join the coalition. In this case, their so-called altruism can be biased and favor those in the coalition.

Differently, under complex endowments, if a blocking coalition $ C $ includes a group of unaffected agents $ C' $ that is self-enforcing, it means that $ C' $ does not rely on the others in $ C $ to ensure their assignments. So their participation in $ C $ is purely altruistic. When $ C' $ has more endowments than their demands, the traditional altruism assumes that $ C' $ is willing to help the agents in $ C $ and harm those outside $ C $. In other words, the pure altruism of $ C' $ is not neutral to the other agents. In Example \ref{Example:Kingdom}, when agent $ 1 $ joins the coalition $ \{1,3\} $ to block $ \sigma $ via $ \delta $, $ 1 $ helps $ 3 $ but harms $ 2 $, while when $ 1 $ joins the coalition $ \{1,2\} $ to block $ \delta $ via $ \sigma $, $ 1 $ helps $ 2 $ but harms $ 3 $. We see that this misused altruism assumption is too permissive towards the formation of blocking coalitions. It also contradicts our intuitive understanding of altruism. Below we rectify the definition of the strong core  by requiring that the altruism of self-enforcing unaffected agents be neutral to the other agents.

\subsection{Our rectification}

To see our idea in the rectification of the strong core, consider a coalition $ C $ that weakly blocks an allocation $ \mu $ via another allocation $ \sigma $. For notational ease, we define
\[
C_{\sigma=\mu}= \{i\in C:\sigma(i)=\mu(i)\} \text{ and }C_{\sigma>\mu}= \{i\in C:\sigma(i)\succ_i\mu(i)\}.
\]

Suppose that a sub-coalition $ C'\subset C_{\sigma=\mu} $ is self-enforcing in $ \sigma $. If $ \w(C')\backslash \sigma(C') $ is nonempty, we call this set the \textbf{redundant endowments} of $ C' $ in $ \sigma $. In the blocking, the agents in $ C' $ help those in $ C\backslash C' $ in three ways:
\begin{enumerate}
	\item If some agents in $ C\backslash C' $ obtain redundant endowments of $ C' $ in $ \mu $ and they wish to reallocate such objects in $ \sigma $, then $ C' $ reallocates such objects as they wish.	
	
	\item If in $ \sigma $ some agents in $ C\backslash C' $ wish to obtain redundant endowments of $ C' $ that are assigned to nobody in $ \mu $, then $ C' $ reallocates such objects as they wish.
	
	\item If in $ \sigma $ some agents in $ C\backslash C' $ wish to obtain redundant endowments of $ C' $ that are assigned to $ I\backslash C $ in $ \mu $, then $ C' $ reallocates such objects as they wish.
\end{enumerate}
The first two ways enhance the efficiency of the allocation of $ \w(C')\backslash \sigma(C') $ and do not harm any agents. But the third way has only a redistribution effect. It cannot be justified by altruism because it helps $ C\backslash C' $ but may harm $ I\backslash C $. 
In our rectification, we assume that every self-enforcing $ C'\subset C_{\sigma=\mu} $ is not willing to reallocate their redundant endowments from $ I\backslash C $  to $ C\backslash C' $. This is formalized as a condition that will be added to Definition \ref{Definition:weakblock}:
\begin{center}
	\textit{For every self-enforcing $ C'\subset C_{\sigma=\mu} $ and $ i\in C\backslash C' $, $ \sigma(i)\in  \omega(C') \implies \sigma(i)\notin \mu(I\backslash C) $.}
\end{center}
In words, if any $ i\in C\backslash C' $ obtains any redundant endowment of $ C' $ in $ \sigma $, then the object cannot be the assignment obtained by any agent outside $ C $ in $ \mu $.

\begin{definition}\label{Definition:induction:block}
	An allocation $ \mu $ is \textbf{rectification blocked} by a coalition $ C $ via another allocation $ \sigma $ if 
	\begin{enumerate}
		\item $ \forall i\in C $, $ \sigma(i) \succsim_i \mu(i) $ and $ \exists j\in C $, $ \sigma(j)\succ_j\mu(j) $;
		
		\item $ \sigma(C)\subset \w(C)\cup \{o^*\} $;
		
		\item for every self-enforcing $ C'\subset C_{\sigma=\mu} $ and $ i\in C\backslash C' $, $ \sigma(i)\in  \omega(C') \implies \sigma(i)\notin \mu(I\backslash C)$.
	\end{enumerate}
The \textbf{rectified core} consists of allocations that are not rectification blocked.
\end{definition}

We will write rectification blocking as \textbf{r-blocking} for short.  
In Example \ref{Example:Kingdom}, agent $ 1 $ will not be willing to join any coalition to block $ \sigma $ or $ \delta $, but $ 1 $ will be willing to join a coalition to block $ \mu $. So the rectified core equals $ \sigma $ and $ \delta $. We present another example to further illustrate the rectified core.

%\begin{remark}
%	Suppose $ \mu $ is r-blocked by a coalition $ C $ via $ \sigma $. If $ C $ is minimal self-enforcing, then the third condition in Definition \ref{Definition:induction:block} is irrelevant. If $ C$ is not minimal self-enforcing, then for every minimal self-enforcing $ C'\subsetneq C $, either $ \mu $ is r-blocked by $ C' $ via $ \sigma $, or for every $ i\in C' $, $ \mu(i)=\sigma(i) $. In the latter case, the third condition will play a role.
%\end{remark}

 \begin{example}
 	Consider five agents $ \{1,2,3,4,5\} $ and five objects $ \{a,b,c,d,e\} $. Agents' endowments and preferences are shown in the following tables. We only list agents' acceptable objects in their preferences. This convention is followed by other examples in the paper.
 	\begin{center}
 		\begin{tabular}{cccccc}
 			& $ a $ & $ b $ & $ c $ & $ d $ & $ e $ \\ \hline
 			$ C_o $: & $ 2 $ & $ 4 $ & $ 1 $ & $ 4 $ & $ 2 $\\ \hline
 			$ \mu $: & $ 2 $ & $ 1 $ & $ 3 $ & $ 4 $ & $ 5 $\\
 			$ \sigma $ & $ 5 $ & $ 1 $ & $ 2 $ & $ 4 $ & $ 3 $
 		\end{tabular}
 		\quad \quad
 		\begin{tabular}{ccccc}
 			$ \succ_1 $	& $ \succ_2 $ & $ \succ_3 $ & $ \succ_4 $ & $ \succ_5 $\\ \hline
 			$ b $ & $ c $ & $ e $ & $ d $ & $ e $\\
 			$ c $ & $ a$ & $ c $ & $ b $ & $ d $\\
 			$ a $ & $ e $ & $ d  $& $ a $ & $ a $
 		\end{tabular}
 	\end{center}
 
  We show that $ \{1,2,3,4\} $ r-blocks $ \mu $ via $ \sigma $. First, the coalition is self-enforcing. Second, in the coalition $ 1,4 $ are unaffected in the blocking. Both $ \{4\} $ and $ \{1,4\} $ are self-enforcing. The redundant endowment of $ \{4\} $ is $ b $, which is assigned to $ 1 $ in both $ \mu $ and $ \sigma $. The redundant endowment of $ \{1,4\} $ is $ c $, which is assigned to $ 2 $ in $ \sigma $ and to $ 3 $ in $ \mu $. So the third condition in Definition \ref{Definition:induction:block} is satisfied. Clearly, $ \sigma $ belongs to the rectified core, because all of $ 1,2,3,4 $ obtain their favorite objects, and $ 5 $ does not own any object.
 \end{example}

The rectified core supersedes the strong core. First, the rectified core is a superset of the strong core, because r-blocking is a restriction of weak blocking. Second, the rectified core is a subset of the weak core, because strong blocking is a special case of r-blocking. Third, the rectified core is Pareto efficient. If an allocation $ \mu $ is Pareto dominated by another allocation $ \sigma $, then $ \mu $ is r-blocked by the grand coalition $ I $ via $ \sigma $. Last, and most importantly, the rectified core is nonempty. 

\begin{theorem}\label{thm:extend core}
	The rectified core is nonempty.
\end{theorem}

Example \ref{Example:Kingdom} has shown that the rectified core can be strictly between the strong core and the weak core. The strong core in Example \ref{Example:Kingdom} is empty. Below, Example \ref{Example:strongcore=subsetof:inductioncore} shows that the rectified core can be a strict superset of the strong core when the strong core is nonempty. Example \ref{Example:co-ownership} shows that the rectified core can be a strict subset of the intersection of the weak core and the set of Pareto efficient allocations.
 
  \begin{example}[Nonempty strong core $ \subsetneq $ rectified core]\label{Example:strongcore=subsetof:inductioncore}
 	This example is a modification of Example \ref{Example:Kingdom}. We add an agent $ 4 $ and a publicly owned object $ c $. Agents' preferences are new.
 	\begin{center}
 		\begin{tabular}{cccc}
 			& $ a $ & $ b $ & $ c $  \\ \hline
 			$ C_o $: & $ 1 $ & $ 1 $ & $ 1,2,3,4 $ \\ \hline
 			$ \mu $: &  & $ 2 $ & $ 1 $ \\
 			$ \sigma $ & $ 1 $ & $ 3 $ & $ 2 $\\
 			$ \delta $ & $ 1 $ & $ 4 $ & $ 2 $ 
 		\end{tabular}
 		\quad \quad
 		\begin{tabular}{cccc}
 			$ \succ_1 $	& $ \succ_2 $ & $ \succ_3 $ & $ \succ_4 $ \\ \hline
 			$ c $ & $ c $ & $ b $ & $ b $\\
 			$ a $ & $ b$ & \\
 			& \\
 			&
 		\end{tabular}
 	\end{center}
 	
 	 Among the three allocations $ \{\mu,\sigma,\delta\} $, only $ \mu $ belongs to the strong core. $ \{1,4\} $ can weakly block $ \sigma $ via $ \delta $ by reallocating $ b $ from $ 3 $ to $ 4 $, and $ \{1,3\} $ can weakly block $ \delta $ via $ \sigma $ by reallocating $ b $ from $ 4 $ to $ 3 $. 
 	 
 	 All of the three allocations belong to the rectified core. $ \{1,4\}  $ cannot r-block $ \sigma $ via $ \delta $ because $ 1 $ is unaffected and self-enforcing and she cannot reallocate $ b $ from $ 3 $ to $ 4 $. For the same reason, $ \{1,3\} $ cannot r-block $ \delta $ via $ \sigma $.
 \end{example}

\begin{example}[Rectified core $ \subsetneq $ (weak core $ \cap $ PE)]\label{Example:co-ownership}
	Consider three agents $ \{1,2,3\} $ and one object $ \{a\}$. The only object is co-owned by $ \{1,2\} $. Every agent prefers $ a $ to $ o^* $.
	\begin{center}
		\begin{tabular}{cc}
			& $ a $   \\ \hline
			$ C_o $: & $ 1,2 $ \\ \hline
			$ \mu $: & $ 1 $  \\
			$ \sigma $ & $ 2 $  \\
			$ \delta $ &  $ 3 $ \\
			$ \eta $ & 
		\end{tabular}
		\quad \quad
		\begin{tabular}{ccc}
			$ \succ_1 $	& $ \succ_2 $ & $ \succ_3 $  \\ \hline
			$ a $ & $ a $ & $ a $ \\
			& \\
			&\\
			&\\
			&
		\end{tabular}
	\end{center}

There are four possible allocations. The weak core equals $ \{\mu,\sigma,\delta,\eta\} $. $ \delta $ and $ \eta $ are not strongly blocked by $ \{1,2\} $ because they cannot be better off simultaneously. The set of Pareto efficient allocations is $ \{\mu,\sigma,\delta\} $.

 The rectified core coincides with the strong core, both being $ \{\mu,\sigma\} $. $ \delta $ is r-blocked by $ \{1,2\} $ via $ \mu/ \sigma $, because one of them will be unaffected and self-enforcing (by obtaining $ o^* $) but she does not have private endowments.

\end{example}

\section{Refinement of the exclusion core }\label{Section:refine:exclusion core}

\subsection{The exclusion core and its inadequacy in some economies}

In BK's interpretation, endowments give their owners the right to evict the others from the objects they directly or indirectly control. The objects directly controlled by a coalition $ C $ are their endowments $ \w(C) $. In an allocation $ \mu $, $ C $ gets indirect control of the endowments of the agents who occupy $ \w(C) $ and, inductively, $ C $ gets the indirect control of the endowments of those who occupy objects directly or indirectly controlled by $ C $. So the set of objects controlled by $ C $ is defined to be $ \Omega(C|\w,\mu)=\w(\cup_{k=0}^\infty C^k) $, where $ C^0=C $ and $ C^k=C^{k-1}\cup (\mu^{-1}\circ \w)(C^{k-1}) $ for every $ k\ge 1 $. $ C $ is said to exclusion block $ \mu $ via another allocation $ \sigma $ if all members of $ C $ are strictly better off in $ \sigma $, and any agent who is worse off in $ \sigma $ is evicted from the objects controlled by $ C $. 

\begin{definition}[\cite{balbuzanov2019endowments}]\label{Definition:exclusion:core}
	An allocation $ \mu $ is \textbf{exclusion blocked} by a coalition $ C $ via another allocation $ \sigma $ if 
	\begin{enumerate}
		\item $ \forall i\in C $, $ \sigma(i) \succ_i \mu(i) $;
		
		\item $\mu(j)\succ_j \sigma(j)\implies \mu(j)\in  \Omega(C|\w,\mu)$.
	\end{enumerate}
	The \textbf{exclusion core} consists of allocations that are not exclusion blocked.
\end{definition}

We will write exclusion blocking as \textbf{e-blocking} for short.

There are two notable features of e-blocking. First, an e-blocking coalition needs not to be self-enforcing. It can happen that a coalition evicts an agent but occupies the agent's endowments in the new allocation (recall the example in Introduction). This cannot happen when agents only have the right to consume or exchange endowments.  A merit of this freedom is that it ensures the exclusion core to be Pareto efficient. If an allocation $ \mu $ is Pareto dominated by another allocation $ \sigma $, then $ \mu $ is e-blocked by $ I_{\sigma>\mu} $ (the set of strictly better off agents in $ \sigma $) via $ \sigma $. The coalition $ I_{\sigma>\mu} $ may not be self-enforcing and $ \Omega( I_{\sigma>\mu} |\w,\mu) $ can be empty. But because no agent is worse off in $ \sigma $, the condition 2 of Definition \ref{Definition:exclusion:core} is satisfied. BK have proved that the exclusion core is nonempty and included by the intersection of the weak core and PE. But the exclusion core is not a superset of the strong core. It can rule out an allocation in the strong core when the latter is nonempty. 

Second, every member in an e-blocking coalition must be strictly better off in the new allocation. BK emphasize that excluding unaffected agents from e-blocking coalitions is necessary for the exclusion core to be nonempty. Otherwise, weak blocking will imply e-blocking, which means that the exclusion core will be a subset of the strong core and thus it will have the same problem with the strong core.\footnote{Suppose that a coalition $ C $ weakly blocks an allocation $ \mu $ via another allocation $ \sigma $. Because $ C $ is self-enforcing, if we define an allocation $ \delta $ where for all $ i\in C $, $ \delta(i)=\sigma(i) $, for all $ j\in I\backslash C $ with $ \mu(j)\in \sigma(C) $, $ \delta(j)=o^* $, and for all remaining $ j' $, $ \delta(j')=\mu(j') $, then every agent who is worse off in $ \sigma $ is evicted from $ \w(C) $.} However, as BK have noticed, excluding unaffected agents makes the exclusion core fail to rule out unintuitive allocations in some economies. We use Examples \ref{Example:Kingdom} and \ref{Example:co-ownership} to illustrate these points.

\begin{customex}{Examples \ref{Example:Kingdom} and \ref{Example:co-ownership} revisited}
	In Example \ref{Example:Kingdom}, suppose we allow e-blocking coalitions to include unaffected agents. Then $ \{1,3\} $ will be able to block $ \sigma $ via $ \delta $  by evicting $ 2 $ from $ b $, and $ \{1,2\} $ will be able to block $ \delta $ via $ \sigma $  by evicting $ 3 $ from $ b $. So the exclusion core will be empty.
	
	On the other hand, when unaffected agents are excluded from e-blocking coalitions, in Example \ref{Example:co-ownership}, the exclusion core equals the intersection of the weak core and PE. So $ \delta $ is in the exclusion core. But $ \delta $ is unintuitive, because the only object $ a $ is not assigned to any of its owners. $ \{1,2\} $ cannot e-block $ \delta $ because they cannot be better off simultaneously. But $ \delta $ is easily ruled out by the strong core and by the rectified core.
\end{customex}

BK ascribe the inadequacy of the exclusion core in Example \ref{Example:co-ownership} to the inflexibility of exclusion rights generated from endowments. BK argue that each of $ \{1,2\} $ should have the right to evict $ 3 $ from their endowment $ a $, but they should not be able to evict each other. So BK introduce relational economies by replacing endowments with priorities to enrich the structure of exclusion rights, and extend the exclusion core.

Instead, we believe that including unaffected agents in e-blocking coalitions is a natural and convenient solution to the above inadequacy of the exclusion core. Our key idea is similar to Section \ref{Section:rectified:core}. To motivate our idea, we distinguish between the two types of unaffected agents in Examples \ref{Example:Kingdom} and \ref{Example:co-ownership}.
 
In Example \ref{Example:Kingdom}, suppose that $ \{1,3\} $ can e-block $ \sigma $ via $ \delta $ by evicting $ 2 $ from $ b $, and $ \{1,2\} $ can block $ \delta $ via $ \sigma $  by evicting $ 3 $ from $ b $. In both coalitions $ 1 $ is self-enforcing and unaffected. By being self-enforcing, it means that $ 1 $ cannot be evicted by any others. So her participation in the two coalitions can be explained by pure altruism. But then we note that, because $ b $ is $ 1 $'s private endowment, both blocking essentially reflects only $ 1 $'s exclusion right. $ 1 $ helps the agent in the blocking coalition but harms the agent not in the coalition. 

Differently, in Example \ref{Example:co-ownership}, suppose that $ \{1,2\} $ can e-block $ \delta $ via $ \sigma $ by evicting $ 3 $ from $ a $. In the coalition $ 1 $ is unaffected and self-enforcing (because $ 1 $ obtains $ o^* $), but because $ \{1,2\} $ co-own $ a $, the blocking reflects their joint exclusion right. So the exclusion core will be $ \{\mu, \sigma\} $, which are the two intuitive allocations in Example \ref{Example:co-ownership}. 
The contrast between the two examples hints at the correct role that unaffected agents should play: unaffected agents should join an e-blocking coalition only to help the others in the coalition use their joint exclusion rights; they should not use their own exclusion rights.

\subsection{Our refinement}

Presenting our idea above, we formally define our refinement of the exclusion core. Recall that a coalition $ C $ is self-enforcing in an allocation $ \sigma $ if $ \sigma(C)\subset \w(C)\cup \{o^*\} $. We further call $ C $ \textbf{minimal self-enforcing} in $ \sigma $ if it is self-enforcing and every sub-coalition of $ C $ is not self-enforcing. 

\begin{definition}\label{Definition:refined}
	An allocation $ \mu $ is \textbf{refined exclusion blocked} by a coalition $ C $ via another allocation $ \sigma $ if 
	\begin{enumerate}
		\item $ \forall i\in C $, $ \sigma(i) \succsim_i \mu(i) $ and $ \exists j\in C $, $ \sigma(j)\succ_j\mu(j) $;

		\item $\mu(j)\succ_j \sigma(j)\implies \mu(j)\in  \Omega(C|\w,\mu)$
		
		\item If $ C_{\sigma=\mu}\neq \emptyset $, then one of the following two cases holds:
		\begin{enumerate}
			\item $ C_{\sigma=\mu} $ is self-enforcing in $ \sigma $ and $\mu(j)\succ_j \sigma(j)\implies \mu(j)\notin \Omega(C_{\sigma=\mu}|\w,\mu) $;
			
			\item $ C $ is minimal self-enforcing in $ \sigma $.
		\end{enumerate}
	\end{enumerate}
	The \textbf{refined exclusion core} consists of allocations that are not refined exclusion blocked.
\end{definition}

It is clear that the refined exclusion core is included by the exclusion core, because e-blocking is a special case of refined e-blocking.  The condition 3 of Definition \ref{Definition:refined} restricts the participation of unaffected agents in a refined e-blocking coalition. We explain the two cases in the condition 3 carefully. 

In case 3(a), the requirement $\mu(j)\succ_j \sigma(j)\implies \mu(j)\notin \Omega(C_{\sigma=\mu}|\w,\mu) $ formalizes our idea that $ C_{\sigma=\mu} $ should not harm any others by using their own exclusion rights. Yet we need to explain why $ C_{\sigma=\mu} $ needs to be self-enforcing when we do not regard endowments as the rights to consume or exchange. When $ C_{\sigma=\mu} $ is self-enforcing, it means that they cannot be evicted by any others. So their participation in $ C $ can be explained by pure altruism. Similar to the idea in Section \ref{Section:rectified:core}, we restrict their altruism by requiring that they do not harm any others by using their own exclusion rights. The power of this restriction has been illustrated by Examples \ref{Example:Kingdom} and \ref{Example:co-ownership}. On the other hand, if $ C_{\sigma=\mu} $ is not self-enforcing, the freedom for $ C $ to evict an agent but occupy her endowment in the new allocation can result in a problem and make the refined exclusion core empty. This is illustrated by Example \ref{Example:why:1}.

\begin{example}[Why $ C_{\sigma=\mu} $ needs to be self-enforcing in condition 3(a)]\label{Example:why:1}
	Consider three agents $ \{1,2,3\} $ and two objects $ \{a,b\} $. Object $ a $ is owned by $ \{1,2\} $, and object $ b $ is privately owned by $ 3 $.
	
	\begin{center}
		\begin{tabular}{ccc}
			& $ a $ & $ b $  \\ \hline
			$ C_o $: & $ 1,2 $ & $ 3 $ \\ \hline
			$ \mu $: & $ 3 $ & $ 1 $ \\
			$ \sigma $: & $ 3 $ & $ 2 $\\
			$ \delta_1 $: &$ 2 $ & $ 1 $\\
			$ \delta_2 $: &$ 1 $ & $ 2 $
		\end{tabular}
		\quad \quad
		\begin{tabular}{ccc}
			$ \succ_1 $	& $ \succ_2 $ & $ \succ_3 $  \\ \hline
			$ b $ & $ b $ & $ a $ \\
			$ a $ & $ a $ & $ b $ \\
			&\\
			&\\
			&
		\end{tabular}
	\end{center}
	
	Suppose we do not require $ C_{\sigma=\mu} $ to be self-enforcing in the condition 3(a) of Definition \ref{Definition:refined}. If the refined exclusion core in this new definition is nonempty, it is still a subset of the exclusion core. So it is Pareto efficient. In this example, because agents accept all objects, a Pareto efficient allocation must assign both objects to agents. Because $ 3 $ privately owns $ b $, if $ 3 $ obtains nothing, $ 3 $ will request $ b $ by evicting any other agent who occupies $ b $. So an unblocked allocation must assign one object to $ 3 $. Finally, because $ 1,2 $ prefer $ b $ to $ a $ while $ 3 $ prefers $ a $ to $ b $, a Pareto efficient and unblocked allocation must assign $ a $ to $ 3 $ and $ b $ to one of $ 1,2 $. So we only need to verify whether the two allocations $ \mu $ and $ \sigma $ belong to the refined exclusion core.
	
	However, $ \mu $ will be refined e-blocked by $ \{1,2\} $ via $ \delta_1 $. $ \{1,2\} $ can evict $ 3 $ from $ a $ because $ a\in \Omega(1,2|\w,\mu)=\{a,b\} $. Because $ 1 $ does not have any private endowments, $ a\notin\Omega(1|\w,\mu)=\emptyset $. Note that $ 1 $ is unaffected in the coalition but she is not self-enforcing. She occupies $ 3 $'s private endowment $ b $. Note also that $ \{1,2\} $ is not self-enforcing, so it does not satisfy the case 3(b). 
	
	Symmetrically, $ \sigma $ will be refined e-blocked by $ \{1,2\} $ via $ \delta_2 $. So the refined exclusion core after dropping the requirement of $ C_{\sigma=\mu} $ being self-enforcing in 3(a) will be empty.
\end{example}

Case 3(b) is the only exception where we do not require $ C_{\sigma=\mu} $ to be self-enforcing, that is, when the coalition $ C $ itself is minimal self-enforcing in $ \sigma $. In this case there are two choices regarding the exclusion rights of $ C_{\sigma=\mu} $. The first choice, which is used by Definition \ref{Definition:refined}, allows $ C_{\sigma=\mu} $ to evict the others by using their own exclusion rights. The second choice does not allow them to do so, and thus it will add the requirement $\mu(j)\succ_j \sigma(j)\implies \mu(j)\notin \Omega(C_{\sigma=\mu}|\w,\mu) $ to 3(b). In the second choice, the condition 3 of Definition \ref{Definition:refined} can be replaced by:

\begin{enumerate}
	\item[\textit{3'.}] \textit{If $ C_{\sigma=\mu}\neq \emptyset $, then $\mu(j)\succ_j \sigma(j)\implies \mu(j)\notin \Omega(C_{\sigma=\mu}|\w,\mu) $, and either $ C_{\sigma=\mu} $ is self-enforcing in $ \sigma $, or $ C $ is minimal self-enforcing in $ \sigma $.}
\end{enumerate}
We present an example to explain the difference between the two choices and why we use the first choice.

\begin{example}[Difference between two choices regarding exclusion rights of $ C_{\sigma=\mu} $ in 3(b)]\label{Example:case2}
	Consider three agents $ \{1,2,3\} $ and two objects $ \{a,b\} $. Agent $ 1 $ privately owns $ b $, and co-owns $ a $ with $ 2 $. All agents accept both objects and prefer $ a $ to $ b $. 	
	\begin{center}
		\begin{tabular}{ccc}
			& $ a $ & $ b $  \\ \hline
			$ C_o $: & $ 1,2 $ & $ 1 $ \\ \hline
			$ \mu $: & $ 1 $ & $ 2 $ \\
			$ \delta $: & $ 1 $ & $ 3 $
		\end{tabular}
		\quad \quad
		\begin{tabular}{ccc}
			$ \succ_1 $	& $ \succ_2 $ & $ \succ_3 $  \\ \hline
			$ a $ & $ a $ & $ a $ \\
			$ b $ & $ b $ & $ b $ \\
			&
		\end{tabular}
	\end{center}

   Consider the two allocations $ \mu $ and $ \delta $. In both allocations $ a $ is assigned to $ 1 $, yet $ b $ is assigned to $ 2 $ in $ \mu $ and assigned to  $ 3 $ in $ \delta $. We argue that $ \delta $ is an unintuitive allocation. Although $ 2 $ is not an owner of $ b $, she at least co-owns $ a $ with $ 1 $, while $ 3 $ owns nothing. But in the allocation $ \delta $, $ 3 $ obtains $ b $ while $ 2 $ obtains nothing. When $ 1 $ obtains $ a $, which is co-owed by $ 2 $, we think that $ 2 $ should be able to ask $ 1 $ to evict $ 3 $ from $ b $.  This is what Definition \ref{Definition:refined} allows.  $ \{1,2\} $ can refined e-block $ \delta $ via $ \mu $, because $ \{1,2\} $ is a minimal self-enforcing coalition.
   
   But if we replace condition 3 by 3', then $ \{1,2\} $ cannot refined e-block $ \delta $ via $ \mu $, because $ b $ is privately owned by the unaffected agent $ 1 $, who is not willing to evict 3 by using her own exclusion right. So $ \delta $ will belong to the refined exclusion core in the new definition.
\end{example}

Example \ref{Example:case2} explains why the first choice is more intuitive than the second. That said, there is no difficulty in our results if we replace condition 3 by 3'. It will make it harder to form blocking coalitions and make the refined exclusion core larger. We prove that the refined exclusion core in Definition \ref{Definition:refined} is nonempty.

\begin{theorem}\label{thm:refined EC}
	The refined exclusion core is nonempty.
\end{theorem}

\section{The ``you request my house -- I get your turn'' mechanism}\label{Section:YRMH}

We propose a generalization of the ``you request my house - I get your turn'' (YRMH-IGYT) mechanism, and prove that all of its outcomes belong to the rectified core and to the refined exclusion core. It means that the two cores are nonempty, and they have a nonempty intersection. This proves Theorem \ref{thm:extend core} and Theorem \ref{thm:refined EC}.

\cite{AbduSonmez1999} propose YRMH-IGYT  to solve the house allocation with existing tenants (HET) model, which is a hybrid of the house allocation model and the housing market model. In the model a subset of objects are privately owned by a subset of agents who are called existing tenants, and the remaining objects are publicly owned by all agents.\footnote{Formally, there is a subset of objects $ O'$ and a subset of agents $ I' $ such that $ |O'|=|I'| $ and every $ i\in I' $ privately owns a distinct object in $ O' $. All objects in $ O\backslash O' $ are publicly owned.} YRMH-IGYT is a generalization of the top trading cycle (TTC) mechanism and the serial dictatorship mechanism.\footnote{In the housing market model, TTC proceeds as follows. In every step, let every agent point to her favorite object and every object point to its owner. Let the agents in each resulting cycle exchange endowments and then remove them. In the house allocation model, the serial dictatorship mechanism proceeds as follows. Choose an order of all agents and then let agents sequentially choose their favorite objects among the remaining.} It proceeds as follows in the HET model. First, fix a linear order of agents, and let private endowments point to their owners. Then, let the first agent in the order receive her favorite object and then remove them, let the second agent in the order receive her favorite object among the remaining ones and then remove them, and so on, until an agent requests an object that is privately owned by an existing tenant. In that case, if the existing tenant has been removed, let the agent in the order receive the object she requests and remove them. But if the existing tenant is not removed, move the existing tenant to the top of the order of the remaining agents, and then proceed as before. If in some step a cycle forms, the cycle must consist of several existing tenants who request each other's private endowments. Let them exchange private endowments as indicated by the cycle, and then remove them.

Our generalization of YRMH-IGYT is different from the original definition in two respects. First, after a cycle is removed in a step, if any agents in the cycle have endowments that have not been removed and any objects in the cycle have owners that have not obtained assignments, in the following steps we treat the remaining endowments of the former agents as owned by the remaining owners of the latter objects. We call this feature \textbf{sharing ownership}. Second, because an object may have several owners in our model, we let every object point to the owner who is ranked highest in the order of agents. After an owner is removed, an object points to the next owner (or a sharing owner) who is ranked highest in the order, if such a (sharing) owner exists. For convenience, we still call our algorithm YRMH-IGYT.

\begin{center}
	\textbf{YRMH-IGYT}
\end{center}

\begin{quote}
	\textbf{Notation}: Fix a linear order of agents $ \rhd $. Let every object point to its owner ranked highest in $ \rhd $.	
	After each step $ t $, let $ I(t) $ denote the set of remaining agents, $ O(t) $ the set of remaining objects, $ C_o(t) $ the set of remaining owners of $ o\in O(t) $, and $ S_o(t) $ the set of agents who obtain shared ownership of $ o\in O(t) $. Let
	$ I(0)=I $, $ O(0)=O $, and for each $ o\in O $, $ C_o(0)=C_o $ and $ S_o(0)=\emptyset $. 
	
	\textbf{Step} $ t\ge 1 $: Maintain remaining arcs, if any, between remaining agents and objects from the previous step. If the agent pointed by any $ o\in O(t-1) $ is removed in the previous step, then:
	\begin{itemize}
		\item If $ C_o(t-1)\neq \emptyset $, let $ o $ point to the $ \rhd $-highest agent in $ C_o(t-1) $;\footnote{Note that although we update the order of agents in the algorithm, the order $ \rhd $ is exogenously fixed.}
		
		\item If $ C_o(t-1)= \emptyset $ and $ S_o(t-1)\neq \emptyset $, let $ o $ point to the $ \rhd $-highest agent in $ S_o(t-1) $;
		
		\item If $ C_o(t-1)= \emptyset $ and $ S_o(t-1)=\emptyset $, let $ o $ not point to any agent.
	\end{itemize} 
	
	Let the highest-ranked agent $ i $ in the current order point to her favorite  object $ o $ among the remaining objects.
	\begin{itemize}
		\item If $ o $ does not point to any agent, or $ o=o^* $, let $ i $ obtain $ o $ and remove them.
		
		\item If $ o $  points to an agent and a cycle forms, let every agent in the cycle obtain the object she points to, and then remove them. For every agent $ j $ in the cycle and every remaining object $ a $ such that $ j\in C_a(t-1)\cup S_a(t-1) $, and for every object $ b $ in the cycle and every remaining agent $ j' $ such that $ j'\in C_b(t-1)\cup S_b(t-1) $, if $ j'\notin C_a(t-1)\cup S_a(t-1) $, let $ j' $ obtain the shared ownership of $ a $ in following steps. 
		
		\item If $ o $  points to an agent but a cycle does not form, move the agent pointed by $ o $ to the top of the current order. 		
	\end{itemize}
    In each step, either an agent is removed, or the order of agents is updated.     
    The algorithm will stop in finite steps, because if no agent is removed in a step, after a few steps either an agent will be removed or a cycle will form.
\end{quote}

By choosing different orders of agents, YRMH-IGYT can find different allocations. We call all of these allocations the \textbf{outcomes} of YRMH-IGYT. We will prove that all of these outcomes belong to the two cores proposed in this paper. 

\begin{theorem}\label{thm:YRMH-IGYT}
	Every outcome of YRMH-IGYT belongs to the rectified core and to the refined exclusion core.
\end{theorem}

We use Example \ref{Example:seconddesign} to illustrate the procedure of YRMH-IGYT. We also use the example to explain why sharing ownership is crucial in our algorithm.

\begin{example}\label{Example:seconddesign}
	Consider four agents $ \{1,2,3, 4\}$ and four objects $ \{a,b,c,d\} $.	
	\begin{center}
		\begin{tabular}{ccccc}
			& $ a $ & $ b $ & $ c $ & $ d $\\ \hline
			$ C_o $: & $ 1,2,3 $ & $ 2 $ & $ 3 $ & $ 4 $\\ \hline
			$ \mu $: & $ 2 $ & $ 4 $ & $ 1 $ & $ 3 $\\
			$ \sigma $: & $ 2 $ & $ 3 $ & $ 1 $ & $ 4 $\\
			&
		\end{tabular}
		\quad \quad
		\begin{tabular}{cccc}
			$ \succ_1 $	& $ \succ_2 $ & $ \succ_3 $ & $ \succ_4 $ \\ \hline
			$ a $ & $ a $ & $ b $ & $ a $\\
			$ b $ & $ c $ & $ d $ & $ b $ \\
			$ c $ & $ b $ & $ c $ & $ d $\\
			$ d $ & $ d $ & $ a $ & $ c $
		\end{tabular}
	\end{center}

	Suppose we choose the order $ 4\rhd 2 \rhd 3 \rhd 1 $. YRMH-IGYT proceeds as follows. Before step one, let $ b,c,d $ point to their private owners, and let $ a $ point to the $ \rhd $-highest owner $ 2 $. In step one, $ 4 $ points to $ a $. Because $ a $ points to $ 2 $, move $ 2 $ to the top of the order. In step two, $ 2 $ points to $ a $ and thus forms a cycle with $ a $. After they are removed, because $ 1 $ and $ 3 $ are remaining owners of $ a $, and $ b $ is $ 2 $'s remaining endowment, $ 1 $ and $ 3 $ obtain the shared ownership of $ b $. In step three, $ 4 $ points to $ b $. Because $ b $ points to $ 3 $, move $ 3 $ to the top of the order. In step four, $ 3 $ points to $ b $ and forms a cycle with $ b $. So they are removed. In step five, $ 4 $ points to $ d $ and forms a cycle with $ d $. So they are removed. In step six, $ 1 $ points to $ c $. Because the owner of $ c $ has been removed, $ 1 $ directly obtains $ c $. So we find the allocation $ \sigma $.
	
	Now suppose we do not have sharing ownership in our algorithm. Then  after we remove the cycle between $ 2 $ and $ a $ in step two, we will not let $ 1 $ and $ 3 $ obtain the shared ownership of $ b $. Then in step three, $ 4 $ will point to $ b $ and directly obtain $ b $. Then the found allocation will be $ \mu $. But $ \mu $ will be rectification blocked and refined exclusion blocked by $ \{1,2,3\}  $ via $ \sigma $.
\end{example}

\section{Relations between solutions}\label{Section:relations}

In this section we clarify the relations between the several solutions we have discussed. For completeness, we also regard YRMH-IGYT as a solution and treat all of its outcomes as the set of allocations it finds for an economy. 
First, Section \ref{Section:rectified:core} has explained that the rectified core includes the strong core and is included by the intersection of the weak core and PE. Second, BK have explained that the exclusion core is included by the intersection of the weak core and PE, but it does not include the strong core. So the exclusion core does not include the rectified core. Third, in Example \ref{Example:co-ownership}, $ \delta $ does not belong to the rectified core, but it belongs to the exclusion core. So the rectified core does not include the exclusion core. Fourth, Example \ref{example:exclusion core not in induction} below shows that the rectified core neither includes the refined exclusion core. So the rectified core and the refined exclusion core do not include each other. Last, Theorem \ref{thm:YRMH-IGYT} proves that YRMH-IGYT is in the intersection of the rectified core and the refined exclusion core. Example \ref{Example:YRMH-IGYT:subset:Refined EC} below shows that the intersection of the two cores can be strictly larger than YRMH-IGYT. So we draw Figure \ref{Figure:relation} to sketch the relations between the several solutions. In Section \ref{Section:special:economies} we will discuss special cases of our model and show that some of these solutions will coincide in these special cases.

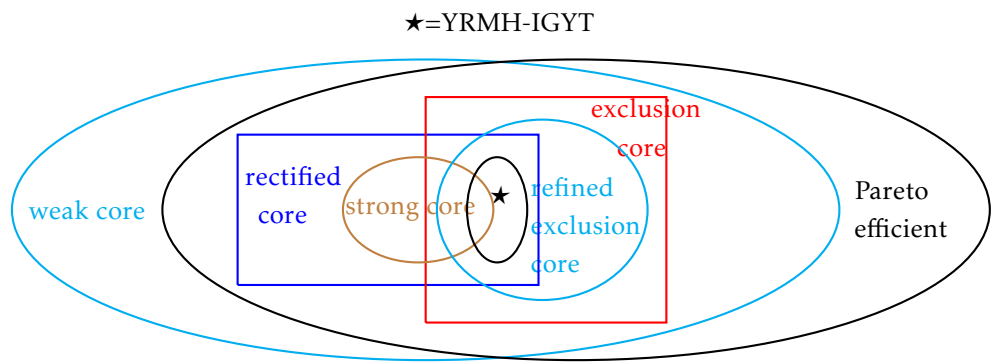
\begin{figure}[!ht]
	\centering
	\begin{tikzpicture}[bend angle=20,xscale=1,yscale=1]
	%\draw[help lines](0,0) grid (8,8);
	\node[cyan]  at (-2.5,2) {\footnotesize weak core};
	\node[black,text width=1cm,align=center]  at (8.2,2) {\footnotesize Pareto\\ efficient};
	\node[blue,text width=1cm,align=center]  at (0.1,2.2) {\footnotesize rectified\\ core};
	\node[brown]  at (1.8,2) {\footnotesize strong core};
	\node[cyan,text width=1cm,align=left]  at (3.9,1.8) {\footnotesize refined\\ exclusion core};
	\node[red,text width=1cm,align=right]  at (4.7,3.1) {\footnotesize exclusion\\ core};
	\node[black] at (3,4.5) {$ \star $\footnotesize =YRMH-IGYT};
	\node[black] at (3,2.2) { $ \star $};
	
	\draw[cyan,thick] (2,2) ellipse (5.5cm and 2cm);
	\draw[black,thick] (4,2) ellipse (5.5cm and 2cm);
	\draw[blue,thick] (-.5,1) --(-.5,3) -- (3.5,3) -- (3.5,1) -- (-.5,1);
	\draw[red,thick] (2,.5) --(5.2,.5) -- (5.2,3.5) -- (2,3.5) -- (2,.5);
	\draw[brown,thick] (1.9,2) ellipse (1cm and .7cm);
	\draw[cyan,thick] (3.55,2) ellipse (1.4cm and 1.2cm);
	\draw[black,thick] (2.95,2) ellipse (.4cm and .7cm);
	%\draw[black,thick,->] (3,2.3) -- (3,4.1);
	\end{tikzpicture}
	\caption{Relations between solutions in the paper.}\label{Figure:relation}
\end{figure}

\begin{example}[Refined exclusion core$ \not\subset $rectified core]\label{example:exclusion core not in induction}
	Consider four agents $ \{1,2,3,4\} $ and three objects $ \{a,b,c\} $.	
	\begin{center}
		\begin{tabular}{cccc}
			& $ a $ & $ b $ & $ c $ \\ \hline
			$ C_o $: & $ 1$ & $ 1,2,3 $ & $ 1,2,3 $\\ \hline
			$ \mu $: & $ 1 $ & $ 2 $ & $ 4 $\\
			$ \sigma $: & $ 1 $ & $ 2 $ & $ 3 $
		\end{tabular}
		\quad \quad
		\begin{tabular}{cccc}
			$ \succ_1 $	& $ \succ_2 $ & $ \succ_3 $ & $ \succ_4 $ \\ \hline
			$ a $ & $ b $ & $ c $ & $ c $ \\
			$ b $ & $ c $ & $ b $ & $ b $ \\
			$ c $& $ a $ & $ a $ & $ a $
		\end{tabular}
	\end{center}
	
	The coalition $ \{1,2,3\} $ rectification blocks $ \mu $ via $ \sigma $. In the coalition $ 1, 2 $ are unaffected, and $ 1 $ is self-enforcing but does not have redundant endowments. However, $ \mu $  belongs to the refined exclusion core. Among the four agents, only $ 3 $ can be better off. Obviously, $ \{1,3\} $ cannot be a blocking coalition because $ 1 $ obtains her only endowment in $ \mu $ while $ 3 $ owns nothing. $ \{2,3,4\} $ and their sub-coalitions can neither be blocking coalitions because they own nothing. So the only possible blocking coalition is $ \{1,2,3\} $. Because $ 1 $ and $ 2 $ have obtained their favorite objects in $ \mu $, they must be unaffected. But $ \{1,2\} $ is not self-enforcing.
\end{example}

\begin{example}[YRMH-IGYT$ \subsetneq $(refined exclusion core$ \cap $rectified core)]\label{Example:YRMH-IGYT:subset:Refined EC}
	Consider four agents $ \{1,2,3,4\} $ and three objects $ \{a,b,c\} $.	
	\begin{center}
		\begin{tabular}{cccc}
			& $ a $ & $ b $ & $ c $ \\ \hline
			$ C_o $: & $ 1,4 $ & $ 1,2,3$ & $ 1,2,3 $\\ \hline
			$ \mu $: & $ 1 $ & $ 2 $ & $ 4 $\\
			$ \sigma $: & $ 1 $ & $ 2 $ & $ 3 $
		\end{tabular}
		\quad \quad
		\begin{tabular}{cccc}
			$ \succ_1 $	& $ \succ_2 $ & $ \succ_3 $ & $ \succ_4 $  \\ \hline
			$ a $ & $ b $ & $ c $ & $ c $ \\
			$ b $ & $ c $ & $ b $ & $ b $\\
			$ c $ & $ a $ & $ a $ & $ a $
		\end{tabular}
	\end{center}
	
	We explain that $ \mu $ cannot be found by YRMH-IGYT, but it belongs to the intersection of the refined exclusion core and the rectified core.
	
	First, suppose that $ \mu $ is an outcome of YRMH-IGYT. Because $ 3 $ most prefers $ c $, $ 3 $ cannot be removed earlier than $ c $. Thus, when $ 4 $ points to $ c $ in some step, because $ 3 $ is an owner of $ c $, $ c $ must point to one of  $ 1,2,3 $. If $ c $ does not point to $ 3 $ in that step, because $ 1 $ obtains $ a $ and $ 2 $ obtains $ b $ in $ \mu $, after they are removed $ c $ will point to $ 3 $. So in any case $ c $ must point to $ 3 $ in some step. But then $ 3 $ will form a cycle with $ c $ and obtain $ c $, which is a contradiction. So $ \mu  $ cannot be an outcome of YRMH-IGYT.
	
	Second, we explain that $ \mu $ is in the rectified core. Because $ 3 $ is the only agent who can be better off, if there exists a rectification blocking coalition, it must include $ 3 $. Because $ \mu $ is Pareto efficient, it cannot be blocked by the grand coalition. Then, it is easy to verify that every other coalition is not self-enforcing.
	
	Third, we explain that $ \mu $ is in the refined exclusion core. As above, any possible blocking coalition must include $ 3 $, and it cannot be the grand coalition. If the coalition is $ \{1,2,3\} $, then $ 4 $ must be evicted from $ c $. But because both the coalition $ \{1,2,3\} $ and the unaffected agents $ \{1,2\} $ are not self-enforcing, the blocking does not hold. If the coalition is $ \{1,3,4\} $, then $ 2 $ must be evicted from $ b $. But $ \mu(2)=b\notin \Omega(\{1,3,4\}|\w,\mu)=\{a\} $, so the blocking neither holds.
\end{example}

\section{Special economies}\label{Section:special:economies}

In this section we evaluate our solutions in special types of economies that generalize the models with simple endowments in the literature. This exercise will deepen our understanding of the relations between the several solutions and provide a new perspective on the models with simple endowments.

We first analyze a type of economies in which agents may co-own objects, but they never own more endowments than their demands, so that every self-enforcing coalition does not have redundant endowments. Formally, we say an economy is \textbf{no-redundant-ownership} if every agent regards all objects as acceptable, and for every coalition $ C $, $ |\w(C)|\le |C| $.\footnote{If some agents regard their endowments as unacceptable, their endowments will become redundant and the situation in Example \ref{Example:Kingdom} can appear.} We will prove that the strong core is nonempty and coincides with the rectified core in this type of economies. This supports our observation that the emptiness of the strong core is caused by the existence of redundant endowments, not by the existence of co-owned endowments. But the strong core and the (refined) exclusion core still may not include each other. Example \ref{example:exclusion core not in induction} belongs to this type. The allocation $ \mu $ in the example is in the refined exclusion core but not in the strong core. Example \ref{Example:HET} on the next page also belongs to this type. As we will see, in the example all allocations in the exclusion core  can be found by YRMH-IGYT, but some allocations in the strong core cannot be found by YRMH-IGYT. 

\begin{proposition}\label{prop:noredundant}
	In no-redundant-ownership economies, YRMH-IGYT$ \subset $strong core=rectified core, but the strong core and the (refined) exclusion core do not include each other.
\end{proposition}

We then analyze a special type of no-redundant-ownership economies that we call \textbf{no-overlapping-ownership}. In every such economy, for every two distinct objects $ a $ and $b $, $ C_a\cap C_b=\emptyset $. So every agent owns at most one object. This type of economies can be regarded as a generalization of the housing market model in which every agent owns one object. For this type, running YRMH-IGYT is equivalent to running TTC by endowing every object only to one of its owners. We prove that YRMH-IGYT can find all allocations in the strong core. The intuition behind is essentially the same as the famous result in the housing market model that TTC finds the unique strong core allocation. We also prove that YRMH-IGYT can find all allocations in the refined exclusion core. But the refined exclusion core can still be a strict subset of the exclusion core. Example \ref{Example:co-ownership} is no-overlapping-ownership. In the example the allocation $ \delta $ is in the exclusion core but not in the refined exclusion core. 

\begin{proposition}\label{prop:nooverlapping}
	In no-overlapping-ownership economies, YRMH-IGYT=strong core= rectified core = refined exclusion core $ \subset  $exclusion core.
\end{proposition}

We turn to another type of economies that we call \textbf{private-ownership}. In every such economy, for every $ o $, $ |C_o|=1 $. So agents may have redundant endowments, but they do not co-own objects. We allow agents to regard some objects as unacceptable. This type of economies can also be regarded as a generalization of the housing market model. But because agents can have redundant endowments, the strong core can be empty (e.g., Example \ref{Example:Kingdom}). For this type, we prove that the rectified core and the exclusion core coincide, and YRMH-IGYT can find all allocations in the two cores. BK have proved that the exclusion core equals the strong core whenever the latter is nonempty. So we have the following the result.

\begin{proposition}\label{prop:private:ownership}
	In private-ownership economies, YRMH-IGYT=rectified core=refined exclusion core=exclusion core, and they equal the strong core when it is nonempty.
\end{proposition}

After considering private ownership, it is natural to consider \textbf{public ownership}. When all objects are publicly owned, BK have proved that the exclusion core equals the strong core, and they both equal the set of Pareto efficient allocations. It is clear that for public-ownership economies, YRMH-IGYT coincides with the serial dictatorship mechanism. So all Pareto efficient allocations are the outcomes of YRMH-IGYT. Then it is immediate to have the following result.

\begin{proposition}\label{prop:public:ownership}
	In public-ownership economies, YRMH-IGYT=strong core=rectified core=refined exclusion core=exclusion core=PE. 
\end{proposition}

Finally, we analyze a hybrid of the above two types that we call \textbf{private-public-ownership}. In such economies, every object is either privately owned or publicly owned. So they can be regarded as an extension of the HET model. Although the several solutions coincide in public-ownership economies, we will show their difference in dealing with public endowments in the presence of private endowments. To motivate our analysis, let us examine the economy in Example \ref{Example:HET}.

\begin{example}\label{Example:HET}
	We first consider an economy of four agents $ \{1,2,3,4\} $ and four objects $ \{a,b,c,d\} $. Object $ d $ is publicly owned. The other objects are privately owned by $ 1,2,3 $ respectively. So it is a HET economy.
	\begin{center}
		\begin{tabular}{ccccc}
			& $ a $ & $ b $ & $ c $ & $ d $ \\ \hline
			$ C_o $: & $ 1 $ & $ 2 $ & $ 3 $ & $ 1,2,3,4 $\\ \hline
			$ \mu $: & $ 1 $ & $ 3 $ & $ 2 $ & $ 4 $\\
			$ \sigma_1 $ & $ 4 $ & $ 1 $ & $ 3 $ & $ 2 $\\
			$ \sigma_2 $ & $ 1 $ & $ 3 $ & $ 4 $ & $ 2 $
		\end{tabular}
		\quad \quad
		\begin{tabular}{cccc}
			$ \succ_1 $	& $ \succ_2 $ & $ \succ_3 $ & $ \succ_4 $ \\ \hline
			$ b $ & $ d $ & $ b $ & $ a $\\
			$ a $ & $ c$ & $ c $ & $ d $\\
			& $ b $ & & $ c $ \\
			&
		\end{tabular}
	\end{center}

In this economy the rectified core coincides with the strong core, both equaling $ \{\mu,\sigma_1,\sigma_2\} $. But the exclusion core coincides with YRMH-IGYT, both equaling $ \{\sigma_1,\sigma_2\} $. The allocation $ \mu $ is exclusion blocked by $ \{1,2,4\} $ via $ \sigma_1 $, and it cannot be found by YRMH-IGYT.

We then consider a second economy that is obtained by adding an agent $ 5 $ to the first economy and letting $ 5 $ privately own $ d $ but most prefer the null object $ o^* $. It means that $ 5 $ does not accept any real object. The other agents' preferences and private endowments remain the same.

\begin{center}
	\begin{tabular}{ccccc}
		& $ a $ & $ b $ & $ c $ & $ d $ \\ \hline
		$ C_o $: & $ 1 $ & $ 2 $ & $ 3 $ & $ 5 $\\ \hline
		$ \mu $: & $ 1 $ & $ 3 $ & $ 2 $ & $ 4 $\\
		$ \sigma_1 $ & $ 4 $ & $ 1 $ & $ 3 $ & $ 2 $\\
		$ \sigma_2 $ & $ 1 $ & $ 3 $ & $ 4 $ & $ 2 $
	\end{tabular}
	\quad \quad
	\begin{tabular}{ccccc}
		$ \succ_1 $	& $ \succ_2 $ & $ \succ_3 $ & $ \succ_4 $ & $ \succ_5 $\\ \hline
		$ b $ & $ d $ & $ b $ & $ a $ & $ o^* $\\
		$ a $ & $ c$ & $ c $ & $ d $\\
		& $ b $ & & $ c $ \\
		&
	\end{tabular}
\end{center}  
	
We argue that the second economy essentially describes the same situation as the first economy. Because in the second economy $ 5 $ most prefers the null object $ o^* $, it is without loss and convenient to ignore $ 5 $ from the economy. After $ 5 $ is removed, her endowment $ d $ becomes unowned. To fit it into our model, it is natural to regard $ d $ as publicly owned by the remaining agents. Then it becomes the first economy.

The first economy is private-public-ownership, while the second economy is private-ownership. Although the two economies essentially describe the same situation, $ \mu $ is not in the rectified core in the second economy: $ \mu $ is rectification blocked by $ \{1,2,4,5\} $ via $ \sigma_1 $. While the exclusion core and YRMH-IGYT still equal $ \{\sigma_1,\sigma_2\} $ in the second economy.
\end{example}

If we imagine public endowments as the leftover private endowments of a leaving agent who demands only the null object $ o^* $, every private-public-ownership economy becomes a private-ownership economy. 
Example \ref{Example:HET} shows that this imagination does not change the exclusion core and the outcomes of YRMH-IGYT, but it makes a difference to the rectified core. Below we formalize the observation from the example as a consistency property of relevant solutions \citep{thomsonconsistency}. 

Let $ \mathcal{E}^{0} $ denote the set of private-public-ownership economies. In every such economy $ \Gamma=(I,O,\succ_I,\{C_o\}_{o\in O}) $, let $ O^P $ denote the set of public endowments. For every $ \Gamma\in \mathcal{E}^{0} $, define an augmented economy $ \Gamma^*=(I\cup \{i^*\},O,\succ_{I\cup \{i^*\}},\{C^*_o\}_{o\in O})  $ by adding an agent $ i^* $ to $ \Gamma $ and letting $ i^* $ privately own $ O^P $ and most prefer $ o^* $. For every $ o\in O\backslash O^p $, let $ C^*_o=C_o $. For every allocation $ \mu $ in $ \Gamma^* $, after removing $ i^* $ and her assignment, we obtain the restriction of $ \mu $ to $ I $, which is an allocation in $ \Gamma $. We denote this restricted allocation in $ \Gamma $ by $ \mu^R $.   

\begin{definition}
	A solution $ f $ is \textbf{consistent} if for every $ \Gamma \in  \mathcal{E}^{0} $, $ f(\Gamma)=\{\mu^R:\mu\in f(\Gamma^*)\} $.
\end{definition}

We prove that the exclusion core and YRMH-IGYT are consistent. The intuition is that, in any $ \Gamma \in  \mathcal{E}^{0} $, if an agent obtains a public endowment, she is impossible to be evicted. So it does not change the exclusion core if we regard public endowments as privately owned by $ i^* $ because $ i^* $ will be unaffected and never join any e-blocking coalitions. While  for any order of agents $ \rhd $, in every step of YRMH-IGYT, after $ i^* $ is removed letting public endowments point to the $ \rhd $-highest agent is equivalent to letting them not point to any agent. So the outcomes of YRMH-IGYT neither change.

\begin{lemma}\label{lemma:consistent}
	The exclusion core and YRMH-IGTY are consistent.
\end{lemma} 

Example \ref{Example:HET} has shown that the rectified core is not consistent. It is easy to prove that the rectified core in every $ \Gamma\in  \mathcal{E}^{0} $ is weakly larger than the rectified core in $ \Gamma^* $.

\begin{lemma}\label{lemma:rectified core}
	For every $ \Gamma \in  \mathcal{E}^{0}$, the rectified core in $ \Gamma\supset $ $ \{\mu^R:\mu\in\text{the rectified core in }\Gamma^*\} $. 
\end{lemma}

For every $ \Gamma \in  \mathcal{E}^{0} $, $ \Gamma^* $ is a private-ownership economy. With Lemma \ref{lemma:consistent} and Lemma \ref{lemma:rectified core}, we can apply Proposition \ref{prop:private:ownership} to conclude that  YRMH-IGYT equals the exclusion core and they are included by the rectified core in every $ \Gamma $. But Example \ref{Example:HET} shows that the exclusion core may not equal the strong core even when the latter is nonempty. Example \ref{Example:strongcore=subsetof:inductioncore} in Section \ref{Section:rectified:core} is private-public-ownership. In the example the rectified core does not equal the strong core when the latter is nonempty. So we have the following result.

\begin{proposition}\label{prop:private:public}
	In private-public-ownership economies, YRMH-IGYT=refined exclusion core =exclusion core$ \subset $rectified core, but they do not equal the strong core.
\end{proposition}

\section{Two applications}\label{Section:application}

We discuss two applications of our idea in this paper. First, we provide a cooperative foundation for YRMH-IGYT in the HET model. In the housing market model, the strong core provides a cooperative foundation for TTC: it is the mechanism that finds the unique allocation in the strong core. But from Example \ref{Example:HET} in Section \ref{Section:special:economies} we see that in the HET model, YRMH-IGYT does not find all allocations in the strong core or in the rectified core. Inspired by the consistency property in Section \ref{Section:special:economies}, we propose a refinement of the rectified core, which will characterize YRMH-IGYT. For every HET economy $ \Gamma$, recall that $ \Gamma^* $ is the augmentation of $ \Gamma $. We regard the rectified core in $ \Gamma^* $ as a solution for $ \Gamma $, and call it \textbf{the rectified core*}. So the rectified core* is consistent. It turns out that this concept has a concise and meaningful definition in the HET model.

\begin{definition}\label{Defn:rectified*}
	In every HET economy, an allocation $ \mu $ is \textbf{rectification blocked*} by a coalition $ C $ via another allocation $ \sigma $ if
	\begin{enumerate}
		\item $ \forall i\in C $, $ \sigma(i) \succsim_i \mu(i) $ and $ \exists j\in C $, $ \sigma(j)\succ_j\mu(j) $;
		
		\item $ \sigma(C)\subset \w(C)\cup \{o^*\}\cup [O^P\backslash \mu(I\backslash C)]  $.
	\end{enumerate} 
	The \textbf{rectified core*} consists of allocations that are not rectification blocked*.
\end{definition}

Condition 2 of Definition \ref{Defn:rectified*} means that the objects that a blocking coalition can reallocate in an allocation include their endowments and public endowments that are not assigned to those outside the coalition. The freedom of reallocating public endowments in the aforementioned case distinguishes the rectified core* from the strong core.

Because the rectified core* is consistent, by Proposition \ref{prop:private:ownership} and Lemma \ref{lemma:consistent} in Section \ref{Section:special:economies}, the rectified core* and the exclusion core both equal YRMH-IGYT in the HET model.

\begin{corollary}
	In HET economies, YRMH-IGYT=rectified core*=exclusion core.	
\end{corollary}

BK claim that the exclusion core characterizes of YRMH-IGYT in the HET model. The rectified core* provides a new characterization of YRMH-IGYT. Comparing with the exclusion core, it uses the standard interpretation of endowments and relaxes the usage of public endowments.

In Section \ref{Section:rectified:core} we have discussed that the problem with the strong core similarly exists in the housing market model when agents' preferences are not strict. Recall the example that three agents $ 1,2,3 $ respectively own three objects $ a,b,c $, and $ 1 $ and $ 3 $ most prefer $ b $ while $ 2 $ is indifferent between all objects. Because $2 $ is willing to exchange endowments with either of the other two agents, this makes the strong core empty.

In our second application, we will apply our idea in Section \ref{Section:rectified:core} to the housing market model with  weak preferences, and propose a new definition of the strong core by modifying the altruism of unaffected agents in blocking coalitions. We will also explore the relation between the strong core in the new definition and extensions of TTC to weak preferences (e.g., extensions proposed by \cite{alcalde2011exchange} and \cite{jaramillo2012difference}). This study will be conducted and presented in a separate paper \cite{zhanghousingmarketcore}.

\section{Related literature}\label{Section:discussion}

Noticing the incompetence of the exclusion core in economies like Example \ref{Example:co-ownership}, BK propose relation economies (where endowments are replaced by priorities) and extend the exclusion core. The strong exclusion core they define may be empty, while the weak exclusion core they define is always nonempty. Because their extensions are defined for a more general model, we do not conduct a thorough comparison between our refinement and theirs. BK propose the Generalized TTC mechanism to find allocations in the weak exclusion core. It is easy to see that our YRMH-IGYT is a special case of their Generalized TTC. For the complex endowments model, the allocation $ \mu $ in Example \ref{Example:seconddesign} can be found by their mechanism, but cannot be found by our YRMH-IGYT. So $ \mu $ is in their weak exclusion core, but not in our refined exclusion core.

\cite{sun2020core} propose another modification of the strong core to solve its nonexistence problem in the complex endowments model. Their solution, called the \textbf{effective core}, can be obtained by replacing the condition 3 of Definition \ref{Definition:induction:block} with the following condition: for every self-enforcing $ C'\subset C $, if $ o\in \w(C')\backslash \sigma(C') $, then $ o\notin \sigma(C) $ unless $ C=I $.
In words, if a blocking coalition $ C $ does not include all agents, then any agent in $ C $ cannot obtain redundant endowments of any self-enforcing sub-coalition of $ C $. It is clear that if this condition is satisfied, our condition 3 is also satisfied. So the effective core is a superset of the rectified core. Sun et al. propose an extension of TTC to find their solution. Our YRMH-IGYT is a special case of their mechanism. Some outcomes of their mechanism cannot be found by YRMH-IGYT, and do not belong to the rectified core.\footnote{An example is available upon request.}

In Section \ref{Section:special:economies}, we define the rectified core* to provide a cooperative foundation for YRMH-IGYT in the HET model.  \cite{ekici2013reclaim} provides a similar but different characterization. In his result, the objects a blocking coalition can reallocate in an allocation include their endowments and assignments. So a coalition may reallocate the others' private endowments, which is not allowed in our solution. \cite{sonmez2010house} characterize YRMH-IGYT in an axiomatic approach, which is different than ours.

Finally, 
\cite{balbuzanov2019property} build the exclusion right into a production network in which agents and firms interact through input-output relations. It seems that the conventional interpretation of endowments cannot be similarly extended. On the other hand, whether the refinement of the exclusion core we propose can apply to production networks is an interesting question. It is left for future research.

%\section{Conclusion}

\appendix

\section{Proof of Theorem \ref{thm:YRMH-IGYT}}

	For any economy $ \Gamma=(I,O,\succ_I,\{C_o\}_{o\in O}) $ and an order of agents $ \rhd $, let $ \mu $ denote the allocation found by YRMH-IGYT. 
	
	\paragraph{$ \mu $ belongs to the rectified core.} Suppose by contradiction that $ \mu $ is rectification blocked by a coalition $ C $ via another allocation $ \sigma $. Among the agents in $ C $ who become strictly better off in $ \sigma $, let $ i_0$ be an agent who is removed earliest in  YRMH-IGYT. Suppose that $ i_0 $ is removed in step $ t_0 $, and $ \sigma(i_0) $ is removed in step $ t $. Because $ \sigma(i_0) \succ_{i_0} \mu(i_0) $,  the object $ o_1= \sigma(i_0) $ must be removed earlier than $ i_0 $. That is, $ t<t_0 $. By the definition of $ i_0 $, all agents in $ C $ who are removed before step $ t_0 $ are indifferent between $ \mu $ and $ \sigma $. Let $ i' $ be the agent who obtains $ o_1 $ in step $ t $. Because $ \mu(i')\neq \sigma(i') $, $ i' $ cannot belong to $ C $. There are two cases.
	
	\medskip
	\noindent \textbf{Case 1}: $ o_1 $ does not point to any agent in step $ t $. So $ C_{o_1}(t-1)= \emptyset $ and $ S_{o_1}(t-1)=\emptyset $. Given that $ C $ has to be self-enforcing in $ \sigma $, because $ i_0\in C $ and $ \sigma(i_0)=o_1 $, $ C_{o_1} $ must belong to $ C $, all owners of $ \sigma(C_{o_1}) $ must belong to $ C $, and so on. So if we define $ C'=\cup_{k=0}^\infty C^k $ where $ C^0= C_{o_1}$ and $ C^k=\cup_{o\in \sigma(C^{k-1})}C_o $ for all $ k\ge 1 $, then $ \sigma(C')\subset \w(C')\cup \{o^*\} $ and $ C'\subset C $. Below we prove by induction that all agents in $ C' $ are removed before step $ t $. 	We explain the first several steps carefully to illustrate our idea.
	
	\textbf{Step 1}: If some $ i_1\in C^0= C_{o_1} $ is not removed before step $ t $, then $ i_1 \in C_{o_1}(t-1)$. But this contradicts $ C_{o_1}(t-1)= \emptyset $. So all agents in $ C_{o_1} $ are removed before step $ t $. It means that for all $ i\in C_{o_1} $, $ \mu(i)=\sigma(i) $.
	
	\textbf{Step 2}: If some $ i_2\in C^1= \cup_{o\in \sigma(C_{o_1})}C_o $ is not removed before step $ t $, then there exist $i_1\in C_{o_1} $ and $ o_2\in \sigma(C_{o_1})$ such that $ \sigma(i_1)=o_2 $ and $ i_2\in C_{o_2} $. When $ i_1 $ obtains $ o_2 $ before step $ t $, because $ i_2 $ remains, $ o_2 $ must point to some agent. This means that $ o_2 $ is involved in a cycle. So after the cycle is cleared, $ i_1 $ must share the ownership of $ o_1 $ with the remaining owners of $ o_2 $. In particular, $ i_2 $ obtains the shared ownership of $ o_2 $. So $ i_2\in S_{o_1}(t-1) $, but it contradicts $ S_{o_1}(t-1)=\emptyset $. So all agents in $ C^1 $ must be removed before step $ t $. Thus, for all $ i\in C^1 $, $ \mu(i)=\sigma(i) $.
	
	\textbf{Step 3}: If some $ i_3\in C^2= \cup_{o\in \sigma(C^1)}C_o $ is not removed before step $ t $, then similarly as above, there exist $i_1\in C_{o_1} $, $ o_2\in \sigma(C_{o_1})$, $ i_2\in C_{o_2} $, and $ o_3\in  \sigma(C^1)$ such that $ \sigma(i_1)=o_2 $, $ \sigma(i_2)=o_3$, and $ i_3\in C_{o_3} $. We denote their relations by a chain:
	\[
	o_1\rightarrow i_1 \rightarrow o_2\rightarrow i_2 \rightarrow o_3\rightarrow i_3.
	\]
	
	We want to prove that $ i_3 $ obtains the shared ownership of $ o_1 $. Because $ i_3 $ is not removed before step $ t $, when $ i_2 $ obtains $ o_3 $, which happens before step $ t $, $ o_3 $ must point to some agent. So $ o_3 $ is involved in a cycle and thus $ i_2 $ must share ownerships with $ i_3 $. If $ o_2 $ is removed before $ i_2 $, then $ o_2 $ must be involved in a cycle when it is removed. So $ i_2 $ must obtain the shared ownership of $ o_1 $, and therefore after $ i_2 $ is removed, $ i_3 $ must obtain the shared ownership of $ o_1 $. If $ o_2 $ is removed after $ i_2 $, then after $ i_2 $ is removed, $ i_3 $ must obtain the shared ownership of $ o_2 $. Then after $ o_2 $ is removed, $ i_1 $ must share the ownership of $ o_1 $ with $ i_3 $. Last, if $ o_2 $ and $ i_2 $ are involved in the same cycle, then after the cycle is removed, $ i_1$ will share the ownership of $ o_1 $ with $ i_3 $ directly. So in every case we must have $ i_3\in S_{o_1}(t-1) $, but it contradicts $ S_{o_1}(t-1)=\emptyset $. So all agents in $ C^2 $ must be removed before step $ t $. Thus, for all $ i\in C^2 $, $ \mu(i)=\sigma(i) $.
	
    \textbf{Step $ \ell\ge 4 $}: If all agents in $ C^{\ell-2} $ are removed before step $ t $, but some $ i_\ell\in C^{\ell-1} $ is not removed before step $ t $. Then there must exist a chain of agents and objects
    \[
    o_1\rightarrow i_1\rightarrow o_2\rightarrow i_2\rightarrow o_3 \rightarrow \cdots \rightarrow o_{\ell-1} \rightarrow i_{\ell-1} \rightarrow o_{\ell}\rightarrow i_\ell 
    \]
    such that for every $ 1\le k\le \ell $, $ i_k\in C_{o_k} $ and $ \sigma(i_k)=o_{k+1} $ ($ \sigma(i_\ell) $ is not defined). Before $ i_\ell $ is not removed before step $ t $, when $ i_{\ell-1} $ and $ o_\ell $ are removed, which happens before step $ t $, $ o_\ell $ must point to some agent and thus be involved in a cycle. So $ i_{\ell-1} $ shares ownerships with $ i_\ell $. By similarly applying the argument in previous steps, the ownership of $ o_1 $ will be passed on along the chain to $ i_\ell$. That is, $ i_\ell\in S_{o_1}(t-1)  $, but this contradicts $ S_{o_1}(t-1)=\emptyset $. So all agents in $ C^{\ell-1} $ are removed before step $ t $.
    
    \textbf{By induction}, every $ i\in C' $ is removed before step $ t $. So for all $ i\in C' $, $ \mu(i)=\sigma(i) $. It means that $ C' $ is self-enforcing and $ C'\subset C_{\sigma=\mu} $. Because $ i_0 $ is strictly better off in $ \sigma $, $ i_0\notin C' $. Then because $ \sigma(i_0)=o_1\in \w(C') $, the third condition in Definition \ref{Definition:induction:block} requires that $ o_1\notin \mu(I\backslash C) $. But it contradicts $ \mu(i')=o_1 $ and $ i'\notin C $.
	
	\medskip
	
    \noindent \textbf{Case 2}: $ o_1 $ is involved in a cycle in step $ t $. Without loss of generality, denote the cycle by
	\[
	i'\rightarrow o_1 \rightarrow i_1 \rightarrow o_2 \rightarrow i_2 \rightarrow \cdots \rightarrow o_\ell \rightarrow i'.
	\]
	Then it must be that $ i_1\in C_{o_1}(t-1) $ or $ i_1\in S_{o_1}(t-1) $. If $ i_1\in C_{o_1}(t-1)  $, then $ i_1\in C $, because $ C_{o_1}\subset C $. If $ i_1\in S_{o_1}(t-1)  $, then all owners of $ o_1 $ must be removed before step $ t $. As in Case 1, we define $ C'=\cup_{k=0}^\infty C^k $ where $ C^0= C_{o_1}$ and $ C^k=\cup_{o\in \sigma(C^{k-1})}C_o $ for all $ k\ge 1 $. Because $ i_1 $ obtains the shared ownership of $ o_1 $ before step $ t $, it must be that $ i_1\in C'\subset C $. Because the cycle is removed in step $ t $, by the definition of $ i_0 $, it must be that $ \sigma(i_1)=\mu(i_1) $. So $ \sigma(i_1)=o_2 $. Because $ C $ is self-enforcing, $ o_2\in \w(C) $. Then either $ i_2\in C_{o_2}(t-1) $ or $ i_2\in S_{o_2}(t-1) $. By applying the above argument to $ i_2 $ and  inductively to the remaining agents in the cycle, we conclude that all agents in the cycle belong to $ C $. But it contradicts $ i'\notin C $.

\paragraph{$ \mu $ belongs to the refined exclusion core.} Because YRMH-IGYT is a special case of BK's Generalized TTC algorithm, $ \mu $ cannot be exclusion blocked.\footnote{BK prove that all outcomes of their Generalized TTC algorithm are in the exclusion core.} Suppose by contradiction that $ \mu $ is refined exclusion blocked by a coalition $ C $ via another allocation $ \sigma $. Because $ \mu $ is not exclusion blocked, $ C_{\sigma=\mu}\neq \emptyset $. Because we have proved that $ \mu $ cannot be rectification blocked, $ C $ is not a minimal self-enforcing coalition.\footnote{If $ C $ is a minimal self-enforcing coalition in $ \sigma $, then $ C $ will satisfy the conditions in Definition \ref{Definition:induction:block}.} So $ C_{\sigma=\mu}  $ must be self-enforcing, and for all $ j $ with $\mu(j)\succ_j \sigma(j)$,  $\mu(j)\in \Omega(C|\w,\mu)\backslash\Omega(C_{\sigma=\mu}|\w,\mu) $. Without loss of generality, let $ C $ include all agents who become strictly better off in $ \sigma $. Note that this does not change the set of unaffected agents $ C_{\sigma=\mu}  $.

 Among the agents who are strictly worse off in $ \sigma $, let $ i_0 $ be an agent who is removed earliest in YRMH-IGYT. Suppose $ i_0 $ is removed in step $ t_0 $. Among the agents who are strictly better off in $ \sigma $, let $ i' $ be an agent who is removed earliest in YRMH-IGYT. Suppose $ i' $ is removed in step $ t_1 $. Because $ \sigma(i') $ is better than $ \mu(i') $ for $ i' $, $ \sigma(i') $ must be removed before step $ t_1 $. Let $ j' $ be the agent who obtains $ \sigma(i') $ in YRMH-IGYT, that is, $ \mu(j')=\sigma(i') $. Then $ j' $ cannot belong to $ C $, because otherwise $ j' $ must be strictly better off in $ \sigma $ and this contradicts the definition of $ i' $. So $ j' $ must be strictly worse off in $ \sigma $. Then $ j' $ cannot be removed earlier than $ i_0 $. It means that $ t_1>t_0 $. Thus, all agents who are strictly better off in $ \sigma $ must be removed after step $ t_0 $, and all agents who are removed before step $ t_0 $ must be indifferent between  $ \sigma $ and $ \mu $.

 Let $ \mu(i_0)= o_1$. Because $ i_0 $ is worse off in $ \sigma $, we need to have $o_1\in \Omega(C|\w,\mu)\backslash\Omega(C_{\sigma=\mu}|\w,\mu) $. But below we prove that if $ o_1\in \Omega(C|\w,\mu)$, then $o_1\in \Omega(C_{\sigma=\mu}|\w,\mu) $. This is a contradiction. We prove this result through proving two lemmas.

\begin{lemma}\label{thm3:lemma1}
	For any agent $ i $ who obtains any object $ o $ in any step $ t\le t_0 $ of YRMH-IGYT, (1) if $ i\in C $, then $ i\in C_{\sigma=\mu} $ and $ o\in \w(C_{\sigma=\mu}) $; (2) if $ o\in \w(C) $, then $ o\in \w(C_{\sigma=\mu}) $;\footnote{Note that with only $ o\in \w(C) $, we cannot conclude that $ i\in C_{\sigma=\mu}  $.} (3) if $ i $ and $ o $ are involved in a cycle in step $ t $ and $ i\in C $ or $ o\in C $, then all agents in the cycle belong to $ C_{\sigma=\mu} $ and all objects in the cycle belong to $ \w(C_{\sigma=\mu}) $.
\end{lemma}

\begin{proof}[Proof of Lemma \ref{thm3:lemma1}]
	We first consider the case that $ i $ and $ o $ are involved in a cycle in step $ t $. Without loss of generality, denote the cycle by
	\[
	i\rightarrow o\rightarrow i_1 \rightarrow o_2 \rightarrow i_2 \rightarrow \cdots \rightarrow o_\ell \rightarrow i.
	\]
	If $ i\in C $, because  all agents in $ C_{\sigma>\mu} $ are removed after step $ t_0 $, it must be that $ i\in  C_{\sigma=\mu}$. Because $ C_{\sigma=\mu} $ is self-enforcing, $ o\in \w(C_{\sigma=\mu}) $. So to prove the third statement, it is sufficient to prove that if $ o\in \w(C) $, then all agents in the cycle belong to $ C_{\sigma=\mu} $ and all objects in the cycle belong to $ \w(C_{\sigma=\mu}) $. There are two cases. If $ i_1 $ is an owner of $ o $, then it must be that $ i_1\in C_{\sigma=\mu} $, because all agents in $ C_{\sigma>\mu} $ are removed after step $ t_0 $. If $ i_1 $ is not an owner of $ o $, then all owners of $ o $ are removed before step $ t $ and $ i_1 $ obtains the shared ownership of $ o $. Because all agents in $ C_{\sigma>\mu} $ are removed after step $ t_0 $, it must be that $ C_o\subset C_{\sigma=\mu} $. Let $ C' $ consist of $ C_o $, the owners of $ \mu(C_{o}) $ (i.e., $ \cup_{o'\in \mu(C_{o})} C_{o'}$), the owners of $ \mu(\cup_{o'\in \mu(C_{o})} C_{o'}) $, and so on. Because $ C_{\sigma=\mu} $ is self-enforcing, it must be that $ i_1\in C'\subset C_{\sigma=\mu} $. So in every case we must have $ i_1\in C_{\sigma=\mu} $. Thus, $ o_2=\mu(i_1)\in \w(C_{\sigma=\mu}) $. Applying the above argument to $ i_2 $ and inductively  to the remaining agents and objects in the cycle, we will conclude that all agents in the cycle belong to $ C_{\sigma=\mu} $ and all objects in the cycle belong to $ \w(C_{\sigma=\mu}) $.
	
   Now suppose that $ o $ does not point to any agent in step $ t $. If $ i\in C $, as proved above, it must be that $ i\in  C_{\sigma=\mu}$ and $ o\in \w(C_{\sigma=\mu}) $. If $ o\in \w(C) $, then all owners of $ o $ are removed before step $ t$. Because all agents in $ C_{\sigma>\mu} $ are removed after step $ t_0 $, all owners of $ o $ belong to $ C_{\sigma=\mu} $. So $ o\in \w(C_{\sigma=\mu}) $.
\end{proof}

Recall that $ \Omega\big(C|\w,\mu\big)=\w(\cup_{k=0}^\infty C^\ell) $ where $ C^0=C $ and $ C^k=C^{k-1}\cup (\mu^{-1}\circ \w)(C^{k-1}) $ for every $ k\ge 1 $.

\begin{lemma}\label{thm3:lemma2}
	For any $ k \ge 1 $ and $ i\in C^k\backslash C^{k-1} $, if $ i $ obtains an object $ o $ in any step $ t\le t_0 $ of YRMH-IGYT, then $ o\in \Omega(C_{\sigma=\mu}|\w,\mu) $ and $ o $ does not point to any agent in step $ t $.
\end{lemma}

\begin{proof}[Proof of Lemma \ref{thm3:lemma2}]
	We prove the lemma by induction. 
	
	\textbf{Step 1}: If $ i\in C^1\backslash C^0 $, then $ o\in \w(C^0)=\w(C) $. Suppose by contradiction that $ o $ points to an agent $ i_1 $ in step $ t $. So $ i $ and $ o $ are involved in a cycle. Then either $ i_1 $ is an owner of $ o $, or $ i_1 $ obtains the shared ownership of $ o $. If  $ i_1 $ is an owner of $ o $, then $ i_1\in C $. By Lemma \ref{thm3:lemma1}, all agents in the cycle belong to $ C_{\sigma=\mu} $, which contradicts $ i\in C^1\backslash C^0 $. If $ i_1 $ obtains the shared ownership of $ o $, then all owners of $ o $ are removed before step $ t $. Because all agents in $ C_{\sigma>\mu} $ are removed after step $ t_0 $, it must be that $ o\in \w(C_{\sigma=\mu}) $. By  Lemma \ref{thm3:lemma1}, all agents in the cycle belong to $ C_{\sigma=\mu} $, which contradicts $ i\in C^1\backslash C^0 $. So it must be that $ o $ does not point to any agent in step $ t $. Thus, all owners of $ o $ are removed before step $ t $, which means that $ o\in \w(C_{\sigma=\mu}) $.
	
	\textbf{Step 2}: If $ i\in C^2\backslash C^1 $, then $ o\in \w(C^1) $ but $ o\notin \w(C) $. Suppose by contradiction that $ o $ points to an agent $ i_1 $ in step $ t $. So $ i $ and $ o $ are involved in a cycle.  If $ i_1 $ is an owner of $ o $, then there are two cases. If $ i_1\in C $, then by  Lemma \ref{thm3:lemma1}, all agents in the cycle belong to $ C_{\sigma=\mu} $, which is a contradiction. If $ i_1\in C^1\backslash C $, by the arguments in Step 1, $ \mu(i_1) $ does not point to any agent in step $ t $. So $ \mu(i_1) $ is not involved in a cycle, which is a contradiction. 
	
	If $ i_1 $ obtains the shared ownership of $ o $, then it means that all owners of $ o $ are removed before step $ t $. Note that $ C_o\subset C^1 $. For every $ j_1\in C_o $, if $ j_1\in C $, by Lemma \ref{thm3:lemma1}, $ j_1\in C_{\sigma=\mu}  $ and $ \mu(j_1)\in \w( C_{\sigma=\mu})$. Because $ C_{\sigma=\mu} $ is self-enforcing, if $ j_1 $ shares the ownership of $ o $ with some agents, those agents must belong to $ C_{\sigma=\mu} $. If $ j_1\in C^1\backslash C $, by Step 1, $ \mu(j_1) $ does not point to any agent and $ \mu(j_1)\in \w( C_{\sigma=\mu}) $. So $ \mu(j_1) $ is not involved in a cycle, and this means that $ j_1 $ does not share the ownership of $ o $ with any other agents. Thus, since $ i_1 $ obtains the shared ownership of $ o $, it must be that $ i_1\in C_{\sigma=\mu} $. But by Lemma \ref{thm3:lemma1}, it means that all agents in the cycle involving $ o $ belong to $ C_{\sigma=\mu} $. This is a contradiction. 
	
	So it must be that $ o $ does not point to any agent in step $ t $. This means that all owners of $ o $ are removed before step $ t $. In this case we have proved that, for every $ j_1\in C_o $, $ \mu(j_1)\in \w( C_{\sigma=\mu}) $. So $ o\in \Omega(C_{\sigma=\mu}|\w,\mu) $.

	\textbf{Step $ \ell\ge 3 $}: If $ i\in C^\ell\backslash C^{\ell-1} $, then $ o\in \w(C^{\ell-1}) $ but $ o\notin \w(C^{\ell-2}) $. Suppose by contradiction that $ o $ points to an agent $ i_1 $ in step $ t $. So $ i $ and $ o $ are involved in a cycle.  If $ i_1 $ is an owner of $ o $, then there are two cases. If $ i_1\in C $, then by  Lemma \ref{thm3:lemma1}, all agents in the cycle belong to $ C_{\sigma=\mu} $, which is a contradiction. If $ i_1\in C^{\ell-1}\backslash C $, then by the arguments in previous steps, $ \mu(i_1) $ does not point to any agent in step $ t $. So $ \mu(i_1) $ is not involved in a cycle, which is a contradiction. 
	
	If $ i_1 $ obtains the shared ownership of $ o $, then all owners of $ o $ are removed before step $ t $. Note that $ C_o\subset C^{\ell-1} $. For every $ j_1\in C_o $, if $ j_1\in C $, by Lemma \ref{thm3:lemma1}, $ j_1\in C_{\sigma=\mu}  $ and $ \mu(j_1)\in \w( C_{\sigma=\mu})$. Because $ C_{\sigma=\mu} $ is self-enforcing, if $ j_1 $ shares the ownership of $ o $ with some agents, those agents must belong to $ C_{\sigma=\mu} $. If $ j_1\in C^{\ell-1}\backslash C $, by the arguments in previous steps, $ \mu(j_1) $ does not point to any agent and $ \mu(j_1)\in \Omega(C_{\sigma=\mu}|\w,\mu)  $. So $ \mu(j_1) $ is not involved in a cycle, and this means that $ j_1 $ does not share the ownership of $ o $ with any other agents. Thus, since $ i_1 $ obtains the shared ownership of $ o $, it must be that $ i_1\in C_{\sigma=\mu} $. But by Lemma \ref{thm3:lemma1}, it means that all agents in the cycle involving $ o $ belong to $ C_{\sigma=\mu} $. This is a contradiction. 
	
	So it must be that $ o $ does not point to any agent in step $ t $. This means that all owners of $ o $ are removed before step $ t $. In this case we have proved that for every $ j_1\in C_o $, $ \mu(j_1)\in \Omega(C_{\sigma=\mu}|\w,\mu)  $. So $ o\in \Omega(C_{\sigma=\mu}|\w,\mu) $.
\end{proof}

Now, if $ o_1\in \Omega(C|\w,\mu)$, there exists $ k\ge 1 $ such that $ o_1\in \w(C^{k-1}) $ and $ i_0\in C^k\backslash C^{k-1} $. Because $ i_0 $ obtains $ o_1 $ in step $ t_0 $, by Lemma \ref{thm3:lemma2}, $ o_1\in \Omega(C_{\sigma=\mu}|\w,\mu) $. This is a contradiction.

\section{Proofs of propositions and lemmas}

\begin{proof}[\textbf{Proof of Proposition \ref{prop:noredundant}}]
    We prove that every allocation $ \mu $ in the rectified core belongs to the strong core. Suppose by contradiction that $ \mu $ is weakly blocked via another allocation $ \sigma $. Let $ C $ be a minimal coalition that weakly blocks $ \mu $ via $ \sigma $. So if $ C' $ is a sub-coalition of $ C $ that is self-enforcing in $ \sigma $, then for all $ i\in C' $, $ \mu(i)=\sigma(i) $. Because $ \sigma(C')\subset \w(C')\cup \{o^*\}  $ and $ |\w(C')|\le |C'| $, if there exists $ j\in C\backslash C' $ such that $ \sigma(j)\in \w(C')\backslash \sigma(C')$, then there exists some $ i_0\in C' $ such that $ \sigma(i_0)=o^* $. Therefore, $ \mu(i_0)=o^* $. But then, because agents accept all objects, $ C' $ can weakly block $ \mu $ via another allocation $ \sigma' $ in which $ \sigma'(i)=\sigma(i) $ for all $ i\in C'\backslash \{i_0\} $ and $ \sigma'(i_0)=\sigma(j) $. But it contradicts the definition of $ C $. So there does not exist $ j\in C\backslash C' $ such that $ \sigma(j)\in \w(C')\backslash \mu(C')$. But this means that $ C $ rectification blocks $ \mu $ via $ \sigma $, because the third condition in Definition \ref{Definition:induction:block} is trivially satisfied. But this is a contradiction. So $ \mu $ belongs to the strong core.
    
    Example \ref{example:exclusion core not in induction} and Example \ref{Example:HET} have explained that the strong core and the (refined) exclusion core may not include each other.
\end{proof}

\begin{proof}[\textbf{Proof of Proposition \ref{prop:nooverlapping}}]
	We first prove that all allocations in the strong core can be found by YRMH-IGYT. Let $ \mu $ be any allocation in the strong core. For every $ o\in O $, there must exist one and only one agent $ i_o\in C_o $ who obtains an object in $ \mu $. Otherwise, there must exist some $ o'\in O $ such that all agents in $ C_{o'} $ obtain nothing in $ \mu $. But then $ C_{o'} $ can weakly block $ \mu $ by allocating $ o' $ to one member of $ C_{o'} $, which is a contradiction. Then $ \mu $ can be found by YRMH-IGYT with an order of agents in which all $ \{i_o\}_{o\in O} $ are ranked above the other agents. 
	
	Similarly, let $ \mu $ be any allocation in the refined exclusion core. For every $ o\in O $, there must exist one and only one agent $ i_o\in C_o $ who obtains an object in $ \mu $. Otherwise, there must exist some $ o'\in O $ such that all agents in $ C_{o'} $ obtain nothing in $ \mu $. But then $ C_{o'} $ can refined exclusion block $ \mu $ by allocating $ o' $ to one member of $ C_{o'} $, which is a contradiction. Note that $ C_{o'} $ is minimal self-enforcing, so the third condition in Definition \ref{Definition:refined} is satisfied. Then $ \mu $ can be found by YRMH-IGYT with an order of agents in which all $ \{i_o\}_{o\in O} $ are ranked above the other agents.
\end{proof}

\begin{proof}[\textbf{Proof of Proposition \ref{prop:private:ownership}}]
	Let $ \mu $ be any allocation in the rectified core or in the exclusion core. We prove that $ \mu $ can be found by YRMH-IGYT for some linear order $ \rhd $ of agents. We will repeat two operations to find such an order. Because the proofs for the two core notions share similar steps, we prove them simultaneously.
	
	Start with the set of all objects and the set of all objects.
	
	\paragraph{Operation A} Let all agents point to their favorite objects and all objects point to their owners. Let the null object $ o^* $ point to every agent. There must exist cycles and cycles must be disjoint. These cycles will appear in YRMH-IGYT with any order of agents. We prove that every agent in every cycle must obtain the object she points to in $ \mu $. Suppose by contradiction that in some cycle not all agents obtain the objects they point to in $ \mu $. Denote by $ C $ the set of the agents in the cycle, by $ C_1 $ the subset of $ C $ who obtain objects they point to in $ \mu $, and by $ C_2 $ the set of the remaining agents in $ C $. Then we argue that $ C_2 $ can exclusion block $ \mu $ via another allocation $ \sigma $ in which, for all $ i\in C $, $ \sigma(i) $ is the object pointed by $ i $, for all $ j\in I\backslash C $ with $ \mu(j)\in \sigma(C) $, $ \sigma(j)=o^* $, and for all other $ j $, $ \sigma(j)=\mu(j) $. First, all agents in $ C_2 $ become strictly better off in $ \sigma $. Second, for any $ j\in I\backslash C  $ with $ \mu(j)\succ_j \sigma(j) $, it must be that $ \mu(j)\in \sigma(C) $. Because the agents in $ C $ form a cycle, it is clear that $ \sigma(C)\subset\Omega(C_2|\w,\mu) $. Similarly, we argue that $ C $ can rectification block $ \mu $ via $ \sigma $. The key observation is that, because the agents in $ C $ form a cycle, $ C $ is minimal self-enforcing in $ \sigma $. But these arguments contradict the assumption that $ \mu $ is in the rectified core or in the exclusion core.
	
	Remove all cycles. Let remaining agents point to their favorite objects among the remaining ones. If there exist cycles, these cycles must be disjoint and will appear in YRMH-IGYT with any order of agents. Remove all cycles. Repeat this operation until no cycles appear. By inductively applying the arguments in the above paragraph, we can prove that all of these cycles must appear in YRMH-IGYT with any order of agents, and the agents in every cycle must obtain the objects they point to in $ \mu $. Denote the set of all removed agents by $ D_1 $. Place $ D_1 $ at the bottom of an order $ \rhd $ and rank them arbitrarily. Then remove $ D_1 $ with their assignments.
	
	\paragraph{Operation B} After Operation A, we will obtain a graph in which all remaining agents point to their favorite remaining objects, all remaining objects point to their owners if their owners are not removed, and otherwise they point to nothing, but there are no cycles. Because every agent points to one agent and every object points to at most one agent, the agents and objects in the graph must form disjoint trees. The root of each tree is an object that points to nothing. Denote these roots by $ o_1,o_2,\ldots,o_m $. Every remaining agent must be connected to one root through a unique directed path in the tree. For every $ \ell\in\{1,\ldots,m\} $, let $ I_\ell $ be the set of agents who directly point to $ o_\ell $. We prove a useful lemma.
	
	\begin{lemma}\label{lemma:operationB}
		There exist $ o_\ell\in \{o_1,o_2,\ldots,o_m\} $ and $ i\in I_\ell $ such that $ \mu(i)=o_\ell $.
	\end{lemma}
    
    \renewcommand{\qedsymbol}{$\blacksquare$}
    
    \begin{proof}[Proof of Lemma \ref{lemma:operationB}]
    	Because $ \mu $ is Pareto efficient, every $ o_\ell $ must be assigned to some agent. Suppose by contradiction that the lemma is not true. That is,  every $ o_\ell\in \{o_1,o_2,\ldots,o_m\} $ is assigned to some agent $ i_\ell $ that is not in $ I_\ell $. Then we prove that $ \mu $ is rectification blocked and also exclusion blocked. We consider two cases.
    	
    	\textbf{Case 1:} If there exists some $i_\ell $ who is connected to $ o_\ell $ through a directed path, without loss of generality, denote the path by
    	\[
    	i_\ell \rightarrow o_1\rightarrow i_1 \rightarrow o_2\rightarrow i_2\rightarrow \cdots \rightarrow i_{k-1}\rightarrow o_k \rightarrow i_k \rightarrow o_\ell.
    	\]In the path, denote by $ C $ the set of agents, by $ C_1 $ the set of agents who obtain the objects they point to in $ \mu $, and by $ C_2 $ the set of remaining agents. It is clear that $ i_\ell,i_k\in C_2 $. Then we argue that $ C_2 $ can exclusion block $ \mu $ via another allocation $ \sigma $ in which, for all $ i\in C $, $ \sigma(i) $ is the object pointed by $ i $, for all $ j\notin C $ with $ \mu(j)\in \sigma(C) $, $ \sigma(j)=o^* $, and for all other $ j $, $ \sigma(j)=\mu(j) $. First, all agents in $ C_2 $ are strictly better off in $ \sigma $. Second, for every $ j $ with $ \mu(j)\succ_j \sigma(j)$, it must be that $ \mu(j)\in \sigma(C) $. Because $ i_\ell,i_k\in C_2 $ and the agents in $ C $ form a chain, it is easy to see that $ \sigma(C)\subset \Omega(C_2|\w,\mu) $.
    	
    	We also argue that $ E=C\cup D_1 $ can rectification block $ \mu $ via $ \sigma $. First, because $ D_1 $ consists of agents who form cycles in Operation A, $ D_1 $ is self-enforcing, and for every $ j\in D_1 $, $ \mu(j)=\sigma(j) $. Second, because the agents in $ C $ form a chain and $ i_k\in C_2 $, for every self-enforcing $ E'\subset E_{\sigma=\mu} $, $ E' $ cannot involve any agent in $ C_1 $. So $ E' $ must be a subset of $ D_1 $. Then if there exists $ i\in E\backslash E' $ such that $ \sigma(i)\in \w(E') $, it must be that $ \sigma(i)=o_\ell $. The third condition in Definition \ref{Definition:induction:block} is satisfied because $ \mu(i_\ell)=o_\ell $. 
    	
    	But the above arguments contradict the assumption that $ \mu $ is in the rectified core or in the exclusion core.
    	
    	\textbf{Case 2:} If there does not exist $ i_\ell $ who is connected to $ o_\ell $, then let $ \{\ell_1,\ell_2, \ldots,\ell_x\} $ be a smallest subset (in set inclusion sense) of $ \{1,2,\ldots,m\} $ such that, for every $ \ell_y\in\{\ell_1,\ell_2, \ldots,\ell_x\} $, $ i_{\ell_y} $ is connected to some $ o_{\ell_z}$ with $ \ell_z\in \{\ell_1,\ell_2, \ldots,\ell_x\} $. Such a subset must exist, because in the worst case $ \{1,2,\ldots,m\} $ is such a set. Denote by $ C $ the union of the set of agents in every directed path that connects every $ i_{\ell_y} $ to $ o_{\ell_z} $, by $ C_1 $ the subset of $ C $ who obtain the objects they point to in $ \mu $, and by $ C_2 $ the set of remaining agents in $ C $. Note that $ C_2 $ is nonempty because every $ i_{\ell_y} $ belongs to $ C_2 $, and for every $ o_{\ell_y} $, the agent who points to $ o_{\ell_y} $ also belongs to $ C_2 $. We argue that $ C_2 $ can exclusion block $ \mu $ via another allocation $ \sigma $ in which, for every $ i\in C $, $ \sigma(i) $ is the object pointed by $ i $, for every $ j\notin C $ with $ \mu(j)\in \sigma(C) $, $ \sigma(j)=o^* $, and for every other $ j $,  $ \sigma(j)=\mu(j) $. First, all agents in $ C_2 $ are strictly better off in $ \sigma $. Second, for every $ j $ with $ \mu(j)\succ_j \sigma(j)$, it must be that $ \mu(j)\in \sigma(C) $. As in Case 1, it is not hard to see that $ \sigma(C)\subset \Omega(C_2|\w,\mu) $.
    	
    	We also argue that $ E=C\cup D_1 $ can rectification block $ \mu $ via $ \sigma $. First, because $ D_1 $ consists of agents who form cycles in Operation A, $ D_1 $ is self-enforcing, and for every $ j\in D_1 $, $ \mu(j)=\sigma(j) $. Second, similarly as in Case 1, for every self-enforcing $ E'\subset E_{\sigma=\mu} $, $ E' $ cannot involve any agent in $ C_1 $. So $ E' $ must be a subset of $ D_1 $. Then if there exists $ i\in E\backslash E' $ such that $ \sigma(i)\in \w(E') $, it must be that $ \sigma(i)=o_{\ell_z} $ for some $ \ell_z\in\{\ell_1,\ell_2, \ldots,\ell_x\} $. The third condition in Definition \ref{Definition:induction:block} is satisfied because every $o_{\ell_z} $ is assigned to some $ i_{\ell_y} $ who belongs to $ E $. 
    	
    	But the above arguments contradict the assumption that $ \mu $ is in the rectified core or in the exclusion core.
    \end{proof}
\renewcommand{\qedsymbol}{$\square$}
	
	Denote the set of $ i\in I_\ell $ who obtains $ o_\ell\in \{o_1,o_2,\ldots,o_m\} $ by $ U_1 $. Place $ U_1 $ at the top of the order $ \rhd $ and rank them arbitrarily. Remove $ U_1 $ with their assignments.
	
	\vspace{.5cm}

	For the remaining agents, repeat Operation A and Operation B. After we obtain the set of agents $ D_k $ in Operation A, place them right above $ D_{k-1} $ in the order $ \rhd $ and rank them arbitrarily. After we obtain the set of agents $ U_k $ in Operation B, place them right below $ U_{k-1} $ in the order $ \rhd $ and rank them arbitrarily. It is easy to verify that $ \mu $ is found by YRMH-IGYT with $ \rhd $. This means that all allocations in the rectified core and all allocations in the exclusion core can be found by YRMH-IGYT. Then the proposition holds.
\end{proof}

\begin{proof}[\textbf{Proof of Lemma \ref{lemma:consistent}}]
	In every $ \Gamma \in  \mathcal{E}^0$, for every allocation $ \mu $, we augment it with an agent $ i^* $ who is assigned $ o^* $ and denote the augmented allocation by $ \mu^A $. Then $ \mu^A $ is an allocation in the augmented economy $ \Gamma^* $.

	We first prove that the exclusion core is consistent. For every $ \Gamma \in  \mathcal{E}^0 $ and every allocation $ \mu $ in the exclusion core in $ \Gamma^* $, we prove that $ \mu^R $ is in the exclusion core in $ \Gamma $. Because $ \mu $ is Pareto efficient in $ \Gamma^* $, $ \mu $  must assign $ o^* $ to $ i^* $ and $ \mu^R $ must be Pareto efficient in $ \Gamma $. Suppose that $ \mu^R $ is exclusion blocked by a coalition $ C $ via another allocation $ \sigma $ in $ \Gamma $. Because $ \mu^R $ is Pareto efficient, it must be that $ C\subsetneq I $. For every $ j $ with $\mu(j)\succ_j \sigma(j)$,  because $\mu(j)\in  \Omega(C|\w,\mu^R)$, $ \mu(j) $ cannot be a public endowment in $ \Gamma $.\footnote{Recall that $ \Omega(C|\w,\mu^R)=\w(\cup_{k=0}^\infty C^k) $, where $ C^0=C $ and $ C^k=C^{k-1}\cup ([\mu^{R}]^{-1}\circ \w)(C^{k-1}) $ for every $ k\ge 1 $. Consider any $ j $ with $\mu(j)\succ_j \sigma(j)$, because $ \mu(j)\in \Omega(C|\w,\mu^R) $, there exists $ k\ge 1 $ such that $ j\in C^k\backslash C^{k-1} $. So $ C^{k-1}\subsetneq I $. It means that $ \w(C^{k-1}) $ cannot include public endowments. Because $ \mu(j)\in \w(C^{k-1}) $, $ \mu(j) $ is not a public endowment.} So $ \mu(j)\in \Omega(C|\w,\mu^R)\backslash O^P $. Let $ \w^* $ be the endowment function in $ \Gamma^* $. Then it is clear that $ \Omega(C|\w^*,\mu)= \Omega(C|\w,\mu^R)\backslash O^P$, because $ O^P $ is the private endowment of $ i^* $ in $ \Gamma^* $ and $ i^*\notin C $. But this means that  in $ \Gamma^* $, $ \mu $ is exclusion blocked by $ C $ via $ \sigma^A $, which is a contradiction. 
	
	Symmetrically, consider any $ \mu $ in the exclusion core in $ \Gamma $. We prove that $ \mu^A $ is in the exclusion core in $ \Gamma^* $. Suppose that $ \mu^A $ is exclusion blocked by a coalition $ C $ via another allocation $ \sigma $ in $ \Gamma^* $. Because $ \mu^A $ is Pareto efficient in $ \Gamma^* $, it must be that $ C\subsetneq I $, and therefore $ \Omega(C|\w,\mu)\supset \Omega(C|\w^*,\mu^A) $. But  this means that in $ \Gamma $, $ \mu $ is exclusion blocked by $ C $ via $ \sigma^R $, which is a contradiction.
	
	We then prove that YRMH-IGYT is consistent. Consider any $ \Gamma \in  \mathcal{E}^0 $ and the procedure of YRMH-IGYT. Fix an order of agents $ \rhd $. In the first step, all public endowments point to the $ \rhd $-highest agent (denoted by $ i_0 $). If $ i_0 $'s favorite object is a public endowment, $ i_0 $ will point to the object and directly obtain it. Otherwise, $ i_0 $ will point to a private endowment (denoted by $ o_1 $) of an agent $ i_1 $. Then $ i_1 $ will be moved to the top of the order. In the next step, if $ i_1 $ points to a public endowment, she will form a cycle with $ i_0 $ and the cycle will be removed. Otherwise, $ i_1 $ will point to a private endowment (denoted by $ o_2 $) of an agent $ i_2 $. Continuing this argument, we will find a cycle of the form
	\[
	o_x\rightarrow i_0 \rightarrow o_1 \rightarrow i_1 \rightarrow o_2 \rightarrow i_2 \rightarrow \cdots \rightarrow i_n\rightarrow o_x
	\]
	where $ o_x $ is a public endowment and every other $ o_\ell $ is privately owned by $ i_\ell $. All agents in the cycle will obtain the objects they point to. 
	
	Now in the procedure of YRMH-IGYT in $ \Gamma^* $, if in the first step $ i_0 $ points to a public endowment as she does in $ \Gamma $, because in $ \Gamma^* $ all public endowments point to $ i^* $, after being moved to the top of the order of agents, $ i^* $ will point to $ o^* $ and be removed with $ o^* $. So in the next step, $ i_0 $ will obtain the public endowment directly as she does in $ \Gamma $. If in the first step $ i_0 $ points to the private endowment $ o_1 $ of $ i_1 $ as she does in $ \Gamma $, then $ i_1 $ will be moved to the top of the order. After $ i_1 $ points to the private endowment $ o_2 $ of $ i_2 $, $ i_2 $ and consequently $ i_3,\ldots,i_n$ will be sequentially moved to the top of the order of agents, until $ i_n $ points to a public endowment $ o_x $. Because $ o_x $ points to $ i^* $, after  $ i^* $ is moved to the top of the order and removed with $ o^* $, $ i_n $ will obtain $ o_x $ directly. After this, the remaining agents in the sequence will sequentially obtain the objects they point to. So the first several steps of YRMH-IGYT are essentially the same in $ \Gamma $ and $ \Gamma^* $.
	
	After $ i_0,i_1,\ldots,i_n $ are removed, we can repeat the above argument to prove that the remaining steps of YRMH-IGYT are essentially the same in $ \Gamma $ and $ \Gamma^* $. This proves that YRMH-IGYT is consistent.
\end{proof}  

\begin{proof}[\textbf{Proof of Lemma \ref{lemma:rectified core}}]
	For every $ \Gamma \in  \mathcal{E}^0 $ and every allocation $ \mu $ in the rectified core in $ \Gamma^* $, we prove that $ \mu^R $ is in the rectified core in $ \Gamma $. Suppose by contradiction that $ \mu^R $ is rectification blocked by a coalition $ C $ via another allocation $ \sigma $. Because $ \mu^R $ is Pareto efficient in $ \Gamma $, $ C\subsetneq I $. So $ \w(C) $ does not include public endowments. Then it is clear that $ C $ can rectification block $ \mu $ via $ \sigma^A $ in $ \Gamma^* $, which is a contradiction.
\end{proof}

\begin{proof}[\textbf{Proof of Proposition \ref{prop:private:public}}]
	For every $ \Gamma \in  \mathcal{E}^0 $, because $ \Gamma^* $ is a private-ownership economy, by Proposition \ref{prop:private:ownership}, YRMH-IGYT equals the exclusion core in $ \Gamma^* $. By Lemma \ref{lemma:consistent}, YRMH-IGYT and the exclusion core are consistent. So YRMH-IGYT also equals the exclusion core in $ \Gamma $. Because the refined exclusion core is always between them, the three solutions coincide. On the other hand, by Lemma \ref{lemma:rectified core}, the rectified core in $ \Gamma $ is a superset of the rectified core in $ \Gamma^* $, which by Proposition \ref{prop:private:ownership} equals the exclusion core. So the rectified core  is a superset of the exclusion core in $ \Gamma$.
\end{proof}

\bibliographystyle{ecta}
\bibliography{reference}

\end{document}